\documentclass[11pt]{article}

\usepackage{tikz}
\usetikzlibrary{matrix,arrows,shapes}
\usepackage{pgf}
\usepackage{url}
  \usepackage{enumerate}
  \usepackage{amsthm}
  \usepackage{amssymb}
  \usepackage{amsmath}
    \usepackage{amscd}
  \usepackage{eucal}
  \usepackage{mathrsfs}
  \usepackage{upgreek}
  \usepackage{mathtools}
   \usepackage{dsfont}
      \usepackage{yfonts}
  \usepackage[T1]{fontenc}
  \usepackage{bbm}
  \usepackage{geometry}
    \usepackage{graphics}
\usepackage{array}
\usepackage{booktabs}
  \usepackage{enumerate}
 \usepackage{feynmf}
  \usetikzlibrary{trees}
\usetikzlibrary{decorations.pathmorphing}
\usetikzlibrary{decorations.markings}
\tikzset{
 photon/.style={decorate, decoration={snake}, draw=red},
    electron/.style={draw=blue, postaction={decorate},
        decoration={markings,mark=at position .55 with {\arrow[draw=blue]{>}}}},
    gluon/.style={decorate, draw=magenta,
        decoration={coil,amplitude=4pt, segment length=5pt}},
    sderiv/.style={postaction={decorate},
        decoration={markings,mark=at position .3 with {\arrow{>}}}},
    tderiv/.style={postaction={decorate},
        decoration={markings,mark=at position .7 with {\arrow{<}}}},
    stderiv/.style={postaction={decorate},
        decoration={markings,mark=at position .7 with {\arrow{<}},mark=at position .3 with {\arrow{>}}}}
}
\newcommand{\E}{\mathfrak{E}}

\newcommand{\fA}{\mathfrak{A}}

\newcommand{\D}{\mathfrak{D}}

\newcommand{\F}{\mathfrak{F}}

\newcommand{\N}{\mathfrak{N}}

\newcommand{\HH}{\mathcal{H}}
\newcommand{\Lcal}{\mathcal {L}}
\newcommand{\Bcal}{\mathcal {B}}
  
\newcommand{\Hcal}{\mathcal{H}}  
\newcommand{\Ccal}{\mathcal{C}}
\newcommand{\Dcal}{\mathcal{D}}
\newcommand{\Ecal}{\mathcal{E}} 
\newcommand{\Fcal}{\mathcal{F}}

\newcommand{\Mcal}{\mathcal{M}}
\newcommand{\Ocal}{\mathcal{O}}
\newcommand{\Scal}{\mathcal{S}}

\newcommand{\Tcal}{\mathcal{T}}

\newcommand{\Ci}{\mathcal{C}^\infty} 
\newcommand{\Loc}{\mathrm{\mathbf{Loc}}}       
\newcommand{\Obs}{\mathrm{\mathbf{Obs}}}       
\newcommand{\Vect}{\mathrm{\mathbf{Vec}}}       

\newcommand{\WF}{\mathrm{WF}}         
\newcommand{\id}{\mathrm{id}}               
\newcommand{\supp}{\mathrm{supp}}      
\newcommand{\dvol}{\mathrm{dvol}_{\sst{M}}} 
\newcommand{\pp}{\mathrm{pp}}
\newcommand{\rp}{\mathrm{rp}}   
\newcommand{\MS}{\mathrm{MS}}    
\newcommand{\loc}{\mathrm{loc}}

\newcommand{\reg}{\mathrm{reg}}
\newcommand{\ren}{\mathrm{r}}

\newcommand{\mc}{{\mu\mathrm{c}}}

\newcommand{\sd}{\mathrm{sd}}

\newcommand{\Diag}{\mathrm{Diag}}
\newcommand{\DIAG}{\mathrm{DIAG}}
\newcommand{\NN}{\mathbb{N}}          
\newcommand{\RR}{\mathbb{R}}           
\newcommand{\CC}{\mathbb{C}}           
\newcommand{\M}{\mathbb{M}} 	     
\newcommand{\al}{\alpha}
\newcommand{\bet}{\beta}

\newcommand{\Ga}{\Gamma}
\newcommand{\de}{\delta}
\newcommand{\De}{\Delta}

\newcommand{\la}{\lambda}

\newcommand{\ph}{\varphi}
\newcommand{\om}{\omega}
\newcommand{\T}{\cdot_{{}^\Tcal}}

\newcommand{\TT}{\Tcal}
\newcommand{\TTR}{\Tcal_\ren}
\newcommand{\TRH}{\cdot_{{}^{\TTR}}}

\newcommand{\scp}[2]{\left\langle#1,#2\right\rangle}
\newcommand{\Poi}[2]{\{#1,#2\}}

\newcommand{\sst}[1]{\scriptscriptstyle{#1}}  
\newcommand{\minus}{\sst{-1}}   
\newcommand{\1}{\mathds{1}}                         
\newcommand{\pa}{\partial}                              
\newcommand{\be}{\begin{equation}}
\newcommand{\ee}{\end{equation}}

\newcommand{\spec}{\mathrm{spec}}
\global\long\def\bld#1{\boldsymbol{#1}}

\begin{document}
\thispagestyle{empty}
\begin{center}
\noindent\rule[2pt]{\textwidth}{2pt}
\textsc{
\Huge Perturbative\\
algebraic quantum field theory\\
\linespread{2}}
\vspace{0.5cm}
\noindent\rule[2pt]{\textwidth}{1pt}\\
\vspace{2cm}
\begin{minipage}{0.4\linewidth}
\flushleft
\Large \textbf{Klaus Fredenhagen}\\
\vspace{2ex}
 \small{ II Inst. f. Theoretische Physik, Universit\"at Hamburg,}\\
    \small{Luruper Chaussee 149,}\\
    \small{D-22761 Hamburg, Germany}\\ 
\small{\texttt{klaus.fredenhagen@desy.de}}\\
\end{minipage}
\hspace{1.5cm}
\begin{minipage}{0.4\linewidth}
\flushleft
\Large \textbf{Katarzyna Rejzner}\\
\vspace{2ex}
 \small{ Department of Mathematics,}\\
 \small{ University of Rome Tor Vergata}\\ 
    \small{Via della Ricerca Scientifica 1,}\\
\small{I-00133, Rome, Italy}\\ 
\small{\texttt{rejzner@mat.uniroma2.it}}\\
\end{minipage}\\
\linespread{2}
\vspace{11cm}
\linespread{2}
\Large \textsc{2012}
\end{center}
\newpage
 \theoremstyle{plain}
  \newtheorem{df}{Definition}[section]
  \newtheorem{thm}[df]{Theorem}
  \newtheorem{prop}[df]{Proposition}
  \newtheorem{cor}[df]{Corollary}
  \newtheorem{lemma}[df]{Lemma}
    \newtheorem{exa}[df]{Example}

  \theoremstyle{plain}
  \newtheorem*{Main}{Main Theorem}
  \newtheorem*{MainT}{Main Technical Theorem}

  \theoremstyle{definition}
  \newtheorem{rem}[df]{Remark}

 \theoremstyle{definition}
  \newtheorem{ass}{\underline{\textit{Assumption}}}[section]

\newpage
\thispagestyle{empty}
$\ $\\
\vspace{6cm}\\
These notes are based on the course given by Klaus Fredenhagen at the Les Houches Winter School in Mathematical Physics (January 29 - February 3, 2012) and the course \textit{QFT for mathematicians} given by Katarzyna Rejzner in Hamburg for the Research Training Group 1670 (February 6 -11, 2012). 
Both courses were meant as an introduction to modern approach to perturbative quantum field theory and are aimed both at mathematicians and physicists.
\clearpage
\tableofcontents
\markboth{Contents}{Contents}
 \section{Introduction}
Quantum field theory (QFT) is at present, by far, the most successful description of fundamental physics. Elementary physics , to a large extent, explained by a specific quantum field theory, the so-called Standard Model. All the essential structures of the standard model are nowadays experimentally verified. Outside of particle physics, quantum field theoretical concepts have been successfully applied also to condensed matter physics.

In spite of its great achievements, quantum field theory also suffers from several longstanding open problems. The most serious problem is the incorporation of gravity.  
For some time, many people believed that such an incorporation would require a radical change in the foundation of the theory. Therefore, theories with rather different structures were favored, for example string theory or loop quantum gravity. But up to now these alternative theories did not really solve the problem; moreover there are several indications that 
QFT might be more relevant to quantum gravity than originally expected.

Another great problem of QFT is the difficulty of constructing interesting examples. In nonrelativistic quantum mechanics, the construction of a selfadjoint Hamiltonian is possible for most cases of interest, in QFT, however, the situation is much worse. Models under mathematical control are
\begin{itemize}
\item free theories,
\item superrenormalizable models in 2 and 3 dimensions,
\item conformal field theories in 2 dimensions,
\item topological theories in 3 dimensions,
\item  integrable theories in 2 dimensions,
\end{itemize}
but no single interacting theory in 4 dimensions is among these models; in particular neither the standard model nor any of its subtheories like QCD or QED. Instead, one has to evaluate the theory in uncontrolled approximations, mainly using formal perturbation theory, and, in the case of QCD, lattice gauge theories.

If one attempts to incorporate gravity, an additional difficulty is the apparent nonlocality of quantum physics which is in conflict with the geometrical interpretation of gravity in Einstein's theory. Even worse, the traditional treatment of QFT is based on several additional nonlocal concepts, including
\begin{itemize}
\item vacuum (defined as the state of lowest energy)
\item particles (defined as irreducible representations of the Poincar\'{e} group)
\item S-matrix (relies on the notion of particles)
\item path integral (involves nonlocal correlations)
\item Euclidean space (does not exist for generic Lorentzian spacetime)
\end{itemize}

There exists, however, a formulation of QFT which is based entirely on local concepts. This is algebraic  quantum field theory (AQFT), or, synonymously, Local Quantum Physics \cite{Haag}. AQFT relies on the algebraic formulation of quantum theory in the sense of the original approach by Born, Heisenberg and Jordan, which was formalized in terms of C*-algebras by I. Segal. The step from quantum mechanics to QFT is performed by incorporating the principle of locality in terms of local algebras of observables. This is the algebraic approach to field theory proposed by Haag and Kastler \cite{HK}. By the Haag-Ruelle scattering theory, the Haag-Kastler framework on Minkowski space, together with some mild assumptions on the energy momentum
spectrum, already implies the existence of scattering states of particles and of the S-matrix. 

It required some time, before this framework could be generalized to generic Lorentzian spacetimes. Dimock \cite{Dim} applied a direct approach, but the framework he proposed did not contain an appropriate notion of covariance. Such a notion, termed {\it local covariance} was introduced more recently in a programmatic paper by Brunetti, Verch and one of us (K.F.) \cite{BFV}, motivated by the attempt to define the renormalized perturbation series of QFT on curved backgrounds \cite{BF0,HW,HW2}. It amounts to an assignment of algebras of observable to generic spacetimes, subject to a certain coherence condition, formulated in the language of category theory. In Section 3 we will describe the framework in detail.

The framework of locally covariant field theory is a plausible system of axioms for a generally covariant field theory. Before we enter the problem of constructing examples of quantum field theory satisfying these axioms, we describe the corresponding structure in classical field theory (Sect. 4). Main ingredient is the so-called Peierls bracket, by which the classical algebra of observables becomes a Poisson algebra.

Quantization can be done in the sense of formal deformation quantization (at least for free field theories), i.e. in terms of formal power series in $\hbar$, and one obtains an abstract algebra resembling the algebra of Wick polynomials on Fock space (Sect. 5). Interactions can then be introduced by the use of a second product in this algebra, namely the time ordered product. Disregarding for a while the notorious UV divergences of QFT we show how interacting theories can be constructed in terms of the free theory (Sect. 6).

In the final chapter of these lecture notes (Sect. 7), we treat the UV divergences and their removal by renormalization. Here, again, the standard methods are nonlocal and loose their applicability on curved spacetimes. Fortunately, there exists a method which is intrinsically local, namely causal perturbation theory. Causal perturbation theory was originally proposed by St\"uckelberg and Bogoliubov and rigorously elaborated by Epstein and Glaser \cite{EG} for theories on Minkowski space. The method was generalized by Brunetti and one of us (K.F) \cite{BF0} to globally hyperbolic spacetimes and was then combined with the principle of local covariance by Hollands and Wald \cite{HW,HW2}. The latter authors were able to show that renormalization can be done  in agreement  with the principle of local covariance. The UV divergences show up in ambiguities in the definition of the time ordered product. These ambiguities are characterized by a group \cite{HW3,DF04,BDF}, namely the renormalization group as originally introduced by Petermann and St\"uckelberg \cite{SP}.   

\section{Algebraic quantum mechanics}
Quantum mechanics in its original formulation in the Dreim\"annerarbeit by Born, Heisenberg and Jordan is based on an identification of observables with elements of a noncommutative involutive complex algebra with unit.
\begin{df}
An involutive complex algebra $\fA$ is an algebra over the field of complex numbers, together with a map, ${}^*:\fA \rightarrow \fA$, called an involution. The image of an element $A$ of $\fA$ under the involution is written $A^*$. Involution is required to have the following properties:
\begin{enumerate}
\item    for all $A, B \in \fA$: $(A + B)^* = A^* + B^*$, $(A B)^* = B^* A^*$,
 \item   for every $\lambda\in\CC$ and every $A \in \fA$: $(\lambda A)^* = \overline{\lambda} A^*$,
\item    for all $A \in \fA$: $(A^*)^* = A$.
\end{enumerate}
\end{df}
In quantum mechanics such an abstract algebra is realized as an operator algebra on some Hilbert space.
\begin{df}
A representation of an involutive unital algebra $\fA$ is a unital ${}^*$-homomorphism $\pi$ into the algebra of linear operators on a dense subspace $\Dcal$ of a Hilbert space $\Hcal$.
\end{df}
Let us recall that an operator $A$ on a Hilbert space $\Hcal$ is defined as a linear map from a subspace $\Dcal\subset\Hcal$ into $\Hcal$. In particular, if $\Dcal=\Hcal$ and $A$ satisfies  $||A||\doteq \sup_{||x||=1}\{||Ax||\} <\infty$, it is called \textit{bounded}. Bounded operators have many nice properties, but
 in physics many important observables are represented by unbounded ones. The notion of an algebra of bounded operators on a Hilbert space can be abstractly phrased in the definition of a $C^*$-algebra.
\begin{df}
A $C^*$-algebra is a Banach  involutive algebra (Banach algebra with involution satisfying $\|A^*\|=\|A\|$), such that the norm has the $C^*$-property:
\[
        \|A^* A \| = \|A\|\|A^*\|, \quad\forall A\in \fA\,.
\]
\end{df}
One can prove that every $C^*$-algebra is isomorphic to a norm closed algebra of bounded operators $\Bcal(\Hcal)$ on a (not necessarily separable) Hilbert space $\Hcal$. A representation of a $C^*$-algebra $\fA$ is a unital ${}^*$-homomorphism $\pi:\fA\rightarrow\Bcal(\Hcal)$. 

In the simplest example from quantum mechanics the algebra of observables is the associative involutive complex unital algebra generated by two hermitian\footnote{An operator $A$ on a Hilbert space $\Hcal$ with a dense domain $D(A)\subset\Hcal$ is called \textit{hermitian} if $D(A)\subset D(A^*)$ and $Ax=A^*x$ for all $x\in D(A)$. It is \textit{selfadjoint} if in addition $D(A^*)\subset D(A)$.} elements $p$ and $q$ with the
canonical commutation relation
\be\label{comut} [p,q]=-i\hbar\1 \ .
\ee
This algebra can be realized as an operator algebra on some Hilbert space, but the operators corresponding to $p$ and $q$ cannot both be bounded. Therefore it is convenient,  to follow the suggestion of Weyl and to replace the unbounded (hence discontinuous) operators $p$ and $q$ by the unitaries\footnote{An element $A$ of an involutive Banach algebra with unit is called unitary if $A^*A=\1=AA^*$.} (Weyl operators) $W(\al,\beta)$, $\al,\beta\in \RR$. Instead of requiring the canonical commutation relation for $p$ and $q$ one requires
the relation (Weyl relation)
\be\label{W:commutator}
W(\al,\beta)W(\al',\beta')=e^{\frac{i\hbar}{2}(\al\beta'-\al'\beta)}W(\al+\al',\beta,\beta')
\ee
The antilinear involution (adjunction)
\be\label{W:involution}
W(\al,\beta)^*=W(-\al,-\beta) \ .
\ee
replaces the hermiticity condition on $p$ and $q$. The Weyl algebra $\fA_W$  is defined as  the unique $C^*$-algebra generated by unitaries $W(\al,\bet)$ satisfying the relations \eqref{W:commutator}, with involution defined by \eqref{W:involution} and with unit $\1=W(0,0)$.

One can show that if the Weyl operators are represented by operators on a Hilbert space such that they depend strongly continuously\footnote{A net $\{T_\al\}$ of operators on a Hilbert space $\Hcal$ converges strongly to an operator $T$ if and only if $||T_\al x-Tx||\rightarrow 0$ for all $x\in\Hcal$.} on the parameters $\al$ and $\beta$,  then $p$ and $q$ can be recovered as selfadjoint generators, i.e.
\[
W(\al,\beta)=e^{i(\al p+\beta q)}\,,
\]
satisfying the canonical commutation relation \eqref{comut}.
As shown by von Neumann, the C*-algebra $\fA_W$ has up to equivalence only one irreducible representation where the Weyl operators depend strongly continuously on their parameters, namely the Schr\"odinger representation $(\mathcal{L}^2(\RR),\pi)$ with
\be
\left(\pi(W(\al,\beta))\Phi\right)(x)=e^{\frac{i\hbar \al\beta}{2}}e^{i\beta x}\Phi(x+\hbar\al)\ ,
\ee
and the reducible representations with the same continuity property are just multiples of the Schr\"odinger representation. If one does not require continuity there are many more representations, and they have found recently some interest in loop quantum gravity. In quantum field theory the uniqueness results do not apply, and one has to deal with a huge class of inequivalent representations. 

For these reasons it is preferable to define the algebra of observables $\fA$ independently of its representation on a specific Hilbert space as a unital C*-algebra. The observables are the selfadjoint elements, and the possible outcomes of measurements are elements of their spectrum. The spectrum  $\spec(A)$ of $A\in\mathfrak{A}$ is the set of all $\lambda\in\CC$ such that $A-\la\1$ has no inverse in $\fA$.  One might suspect that the spectrum could become smaller if the algebra is embedded in a larger one. Fortunately this is not the case; for physics this mathematical result has the satisfactory effect that the set of possible measurement results of an observable is not influenced by the  inclusion of additional observables. 

Now we know what the possible outcome of an experiment could be, but what concrete value do we get, if we perform a measurement?
In QM this is not the right question to ask. Instead, we can only determine the \textit{probability distribution} of getting particular values from a measurement of an observable $A$. This probability distribution can be obtained, if we know the \textit{state} of our physical system. Conceptually, a state is a prescription for the preparation of a system. This concept entails in particular that experiments can be reproduced and is therefore equivalent to the ensemble interpretation where the statements of the theory apply to the ensemble of equally prepared
systems.

A notion of a state can be also defined abstractly, in the following way:
\begin{df}
A state on an involutive algebra $\fA$ is a linear functional $\omega:\fA\rightarrow \CC$, such that:
\[ \omega(A^*A) \geq 0,\qquad\textrm{and}\  \omega(\1) = 1\]
\end{df}
The first condition can be understood as a positivity condition and the second one is the normalization. The values $\om(A)$ are interpreted as the expectation values of the observable $A$ in the given state. Given an observable $A$ and a state $\omega$ on  a C*-algebra $\fA$ we can reconstruct the full probability distribution $\mu_{A,\omega}$ of measured values of $A$ in the state $\omega$ from its moments, i.e. the expectation values of powers of $A$,
\[
\int \lambda^n d\mu_{A,\omega}(\lambda)= \omega(A^n)\,.
\]
States on $C^*$-algebras are closely related to representations on Hilbert spaces. This is provided by the famous GNS (Gelfand-Naimark-Segal) theorem:
\begin{thm}
Let $\omega$ be a state on the involutive unital algebra $\fA$. Then there exists a representation $\pi$ of the algebra by linear operators on a dense subspace $\Dcal$ of some Hilbert space $\HH$ and a unit vector $\Omega \in \Dcal$, such that
\[
\omega(A) = (\Omega, \pi(A)\Omega)\, ,
\]
and $\Dcal = \{\pi(A)\Omega, A \in \fA\}$.
\end{thm}
\begin{proof}
The proof is quite simple. First let us  introduce a scalar product on the algebra in terms of the state $\omega$ by
\[
\scp{A}{B}\doteq\omega(A^*B)\,.
\]
Linearity for the right and antilinearity for the left factor are obvious, hermiticity $\scp{A}{B}=\overline{\scp{B}{A}}$ follows from the positivity of $\omega$ and the fact that we can write $A^*B$ and $B^*A$ as linear combinations
 of positive elements:
 \begin{align*}
  2(A^*B + B^*A) &= (A + B)^*(A + B) - (A - B)^*(A - B) \,,\\
2(A^*B - B^*A) &= -i(A + iB)^*(A + iB) + i(A - iB)^*(A - iB)\,.
  \end{align*}
  Furthermore, positivity of $\omega$ immediately implies that the scalar product is positive semidefinite, i.e. $\scp{A}{A}\geq 0$ for all $A\in\fA$. We now study the set
  \[
  \N\doteq\{A\in\fA|\omega(A^*A)=0\}\,.
  \]
  We show that $\N$ is a left ideal of $\fA$. Because of the Cauchy-Schwarz inequality $\N$ is a subspace of $\fA$. Moreover, for $A \in \N$ and $B \in \fA$ we have, again because of the Cauchy-Schwarz inequality:
\[
\omega((BA)^*BA) = \omega(A^*B^*BA) = \scp{B^*BA}{A} \leq \sqrt{\scp{B^*BA}{ B^*BA}}\sqrt{\scp{A}{A}} = 0\,,
\]
hence $BA\in\N$. Now we define $\Dcal$ to be the quotient $\fA/\N$. Per constructionem the scalar product is positive definite on $\Dcal$, thus we can complete it to obtain a Hilbert space $\HH$. The representation $\pi$ is induced by left multiplication of the algebra,
\[
\pi(A)(B + \N) \doteq AB + \N\,,
\]
and we set $\Omega=\1+\N$.
In case that $\fA$ is a C*-algebra, one can show that the operators $\pi(A)$ are bounded, hence admitting unique continuous extensions to bounded operators on $\HH$. 
 
It is also straightforward to see that the construction is unique up to unitary equivalence. Let $(\pi',\Dcal',\HH',\Omega')$ be another quadruple satisfying the conditions of the theorem. Then we define an operator $U : \Dcal\rightarrow\Dcal'$ by
\[
U\pi(A)\Omega \doteq \pi'(A)\Omega'.
\]
$U$ is well defined, since $\pi(A)\Omega = 0$ if and only if $\omega(A^*A) = 0$, but then we have also $\pi'(A)\Omega' = 0$. Furthermore $U$ preserves the scalar product and is invertible and has therefore a unique extension to a unitary operator from $\HH$ to $\HH'$. This shows that $\pi$ and $\pi'$ are unitarily equivalent.
\end{proof}

%
The representation $\pi$ will not be irreducible, in general, i.e. there may exist a nontrivial closed invariant subspace. In this case, the state $\om$ is not pure, which means that it is a convex combination of other states,
\be
\om=\la\om_1+(1-\la)\om_2\ ,\ 0<\la<1\ ,\ \om_1\not=\om_2\ .
\ee
To illustrate the concept of the GNS representation, let $\pi_{1,2}$ be representations of $\mathfrak{A}$ on Hilbert spaces $\HH_{1,2}$, respectively. Choose unit vectors $\Psi_1\in\HH_1$, $\Psi_2\in\HH_2$ and define the states 
\be
\om_i(A)=\scp{\Psi_i}{\pi_i(A)\Psi_i}\ ,\ i=1,2\ .
\ee
Let $\om$ be the convex combination
\be
\om(A)=\frac12\om_1(A)+\frac12\om_2(A)\ .
\ee
$\om$ is a linear functional satisfying the normalization and positivity conditions and therefore is again a state in the algebraic sense.
Now let $\HH=\HH_1\oplus\HH_2$ be the direct sum of the two Hilbert spaces and let 
\be
\pi(A)=\left(\begin{array}{cc}
                    \pi_1(A) & 0\\
                    0              & \pi_2(A)
                    \end{array}\right)
\ee
Then the vector 
\be
\Psi=\frac{1}{\sqrt{2}}\left(\begin{array}{c}
                                             \Psi_1\\
                                             \Psi_2                                  
                                              \end{array}\right)
\ee
satisfies the required relation
\be
\om(A)=\scp{\Psi}{\pi(A)\Psi} \ .
\ee
For more information on operator algebras see \cite{Reed,BR1,BR2}.

In classical mechanics one has a similar structure. Here the algebra of observables is commutative and can be identified with the algebra of continuous functions on phase space. In addition, there is a second product, the Poisson bracket. This product is only densely defined. States are probability measures, and pure states correspond to the evaluation of functions at a given point of phase space.
\section{Locally covariant field theory}
Field theory involves infinitely many degrees of freedom, associated to the points of spacetime. Crucial for the success of field theory is a principle which regulates the way these degrees of freedom influence each other. This is the principle of locality, more precisely expressed by the German word {\it Nahwirkungsprinzip}. It states that each degree of freedom is influenced only by a relatively small number of other degrees of freedom. This induces a concept of neighborhoods in the set of degrees of freedom.

The original motivation for developing QFT was to combine the QM with special relativity. In this sense we expect to have in our theory some notion of \textit{causality}. Let us briefly describe what it means in mathematical terms. In special relativity space and time are described together with one object, called \textit{Minkowski spacetime}. Since it will be useful later on, we define now a general notion of a \textit{spacetime} in physics.
\begin{df}
A spacetime $(M,g)$ is a smooth (4 dimensional) manifold (Hausdorff, paracompact, connected) with a smooth pseudo-Riemannian metric \footnote{a smooth tensor field $g\in\Gamma(T^*M\otimes T^*M)$, s.t. for every $p\in M$, $g_p$ is a symmetric non degenerate bilinear form.} of Lorentz signature (we choose the convention $(+,-,-,-)$).
\end{df}
A spacetime $M$ is said to be orientable if there exists a differential form of maximal degree (a volume form), which does not vanish anywhere. We say that $M$ is time-orientable if there exists a smooth vector field $u$ on $M$ such that for ever $p\in M$ it holds $g(u,u)>0$. We will always assume that our spacetimes are orientable and time-orientable. We fix the orientation and choose the time-orientation by selecting a specific vector field $u$ with the above property.  Let $\gamma: \RR\supset I\rightarrow M$ be a smooth curve in $M$, for $I$ an interval in $\RR$. We say that $\gamma$ is causal (timelike) if it holds $g(\dot{\gamma},\dot{\gamma})\geq 0$ ($>0$), where $\dot{\gamma}$ is the vector tangent to the curve.

Given the global timelike vectorfield
$u$ on $M$, one calls a causal curve $\gamma$ future-directed if
$g(u,\dot{\gamma}) > 0$ all along $\gamma$, and analogously one
calls $\gamma$ past-directed if $g(u,\dot{\gamma}) < 0$. This
induces a notion of time-direction in the spacetime
$(M,g)$. For any point $p \in M$, $J^{\pm}(p)$ denotes the set of
all points in $M$ which can be connected to $x$ by a
future(+)/past$(-)$-directed causal curve $\gamma: I \to M$ so that $x
= \gamma(\inf \,I)$. The set $J^+(p)$ is called the causal future and 
$J^-(p)$ the causal past of $p$. The boundaries $\partial J^\pm(p)$ of these regions 
are called respectively: the \textit{future}/\textit{past lightcone}.
 Two subsets $O_1$ and $O_2$ in $M$ are called
causally separated if they cannot be connected by a causal curve, i.e.\
if for all $x \in \overline{O_1}$,
 $J^{\pm}(x)$ has empty intersection with $\overline{O_2}$.
By $O^{\perp}$ we denote the causal complement of $O$, i.e.\ the
largest open set in $M$ which is causally separated from $O$.
\begin{figure}[!htb]
\begin{center}
\includegraphics[width=7cm]{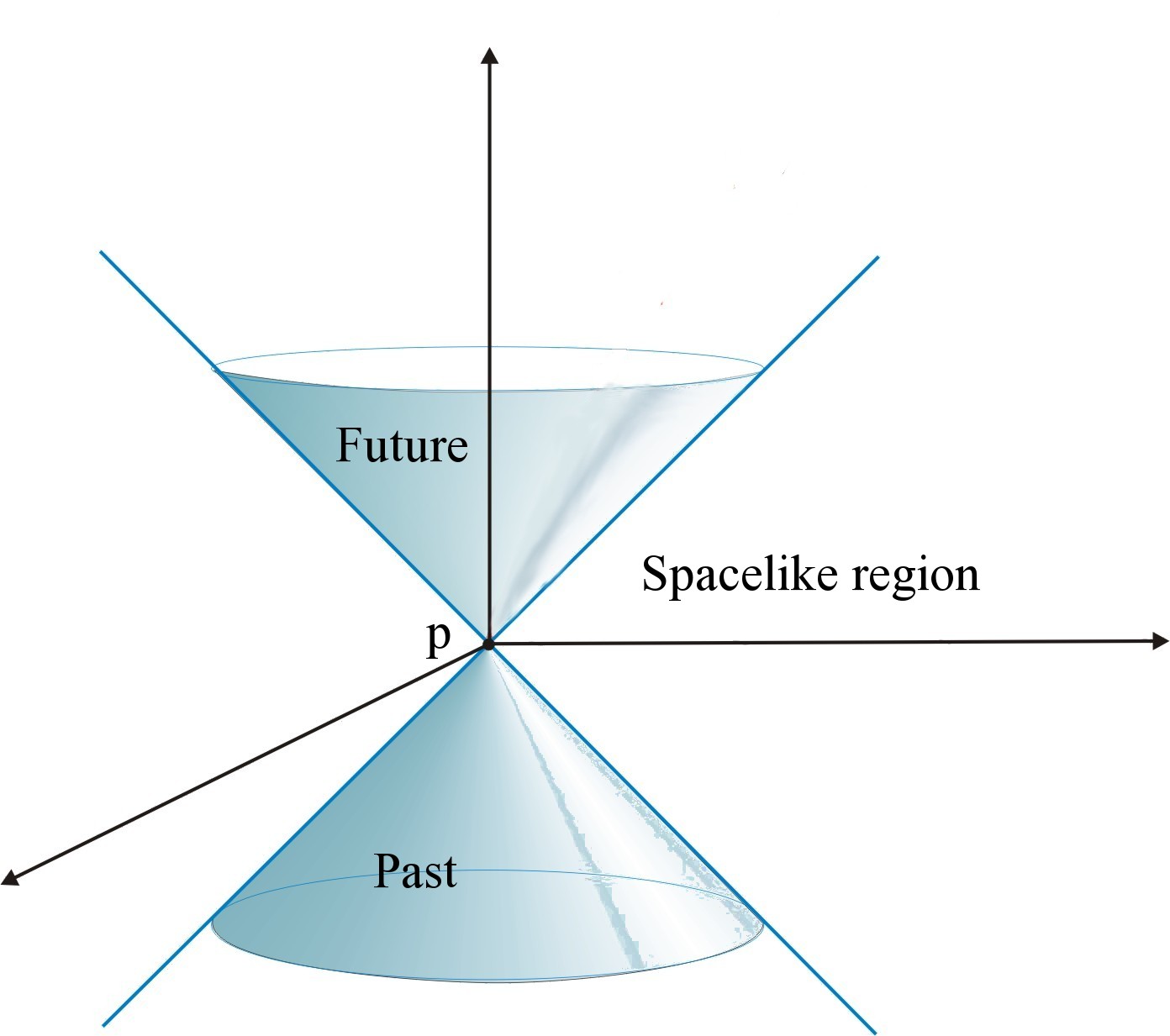}
\end{center}
\caption{A lightcone in Minkowski spacetime.\label{lightcone}}
\end{figure}

In the context of general relativity we will also make use of following definitions:
\begin{df}
A causal curve is \textbf{future inextendible} if there is no $p\in M$ such that: 
\[
\forall U\subset M\mbox{open neighborhoods of }p,\, \exists t'\textrm{ s.t. }\gamma(t)\in U\forall t>t'\,.
\]
\end{df}
\begin{df}
\textbf{A Cauchy hypersurface} in $M$ is a smooth subspace of $M$ such that every inextendible causal curve intersects it exactly once.
\end{df}
\begin{df}
An oriented and time-oriented spacetime $M$ is called \textbf{globally hyperbolic} if there exists a smooth foliation of $M$ by Cauchy hypersurfaces.
\end{df}

For now let us consider a simple case of the \textit{Minkowski spacetime} $\M$ which is just $\RR^4$ with the diagonal metric $\eta=\textrm{diag}(1,-1,-1,-1)$. A \textit{lightcone} with apex $p$ is shown on figure \ref{lightcone}, together with the future and past of $p$.

One of the main principles of special relativity tells us that physical systems which are located in causally disjoint regions should in some sense be independent. Here we come to the important problem: \textit{How to implement this principle in quantum theory?}
A natural answer to this question is provided by the Haag-Kastler framework \cite{Haag0,HK}, which is based on the principle of \textit{locality}.
In the previous section we argued that operator algebras are a natural framework for quantum physics. Locality can be realized by identifying the algebras of observables that can be measured in given bounded regions of spacetime. In other words we associate to each bounded $\Ocal\subset M$ a $C^*$-algebra $\fA(\Ocal)$. This association has to be compatible with a physical notion of subsystems. It means that if we have a region $\Ocal$ which lies inside $\Ocal'$ we want the corresponding algebra $\fA(\Ocal)$ to be contained inside $\fA(\Ocal')$, i.e. in a bigger region we have more observables. This property can be formulated as the \textbf{Isotony} condition for the net $\{\fA(\Ocal)\}$ of local algebras associated to bounded regions of the spacetime. In the Haag-Kastler framework one specializes to Minkowski space $\M$ and imposes some further, physically motivated, properties:
\begin{itemize}
	\item {\bf Locality (Einstein causality)}. Algebras associated to spacelike separated regions commute:
$\Ocal_1$ spacelike separated from $\Ocal_2$, then $[A, B] = 0$, $\forall A \in \fA(\Ocal_1)$, $B \in \fA(\Ocal_2)$. This expresses the ``independence'' of physical systems associated to regions  $\Ocal_1$ and $\Ocal_2$.
	\item {\bf Covariance}. The Minkowski spacetime has a large group of isometrics, namely the Poincar\'e group. We require that  there exists a family of isomorphisms $\alpha_L^\Ocal : \fA(\Ocal) \rightarrow\fA(L \Ocal)$ for Poincar\'e transformations $L$, such that for
$\Ocal_1 \subset \Ocal_2$ the restriction of $\alpha_L^{\Ocal_2}$ to $\fA(\Ocal_1)$ coincides with $\alpha_L^{\Ocal_1}$ and such that: $\alpha^{L\Ocal}_{L'}\circ \alpha^{\Ocal}_{L}=\alpha^{\Ocal}_{L'L}$,
	\item {\bf Time slice axiom}: the algebra of a neighborhood of a Cauchy surface of a given region coincides with the algebra of the full region. Physically this correspond to a well-posedness of an initial value problem. We need to determine our observables in some small time interval $(t_0-\epsilon,t_0+\epsilon)$ to reconstruct the full algebra.
	\item {\bf Spectrum condition}. This condition corresponds physically to the positivity of energy. One assumes that there exist a compatible family of faithful representations $\pi_{\Ocal}$ of 
	$\fA(\Ocal)$ on a fixed Hilbert space (i.e. the restriction of $\pi_{\Ocal_2}$ to $\fA(\Ocal_1)$ coincides with $\pi_{\Ocal_1}$ for $\Ocal_1\subset\Ocal_2$) such that translations are unitarily implemented, i.e. there is a unitary representation $U$ of the translation group satisfying
	\[U(a)\pi_{\Ocal}(A)U(a)^{-1}=\pi_{\Ocal+a}(\alpha_a(A))\ ,\ A\in\fA(\Ocal),\]
	and such that the joint spectrum of the generators $P_{\mu}$ of translations $e^{iaP}=U(a)$, $aP=a^\mu P_\mu$, is contained in the forward lightcone: $\sigma(P)\subset \overline{V}_+$. 
	\end{itemize}

We now want to generalize this framework to theories on generic spacetimes. To start with, we may think of a globally hyperbolic neighborhood  $U$ of a spacetime point $x$ in some spacetime  $M$. 
Moreover, we assume that any causal curve in $M$ with end points in $U$ lies entirely 
in $U$. Then we require that the structure of the algebra of observables associated to $U$ should be completely independent of the world outside. We may formalize this idea by requiring that for any embedding $\chi:M\to N$ of a globally hyperbolic manifold $M$ into another one $N$ which preserves the metric, the orientations and the causal structure\footnote{The property of \textit{causality preserving} is defined as follows: let $\chi:M\to N$, for any 
			causal curve $\gamma : [a,b]\to N$, if $\gamma(a),\gamma(b)\in\chi(M)$ then for all $t	\in ]a,b[$ we have: $\gamma(t)\in\chi(M)$ .}
(these embeddings will be called {\it admissible}), there exist an injective homomorphism 
\be
\alpha_{\chi}:\mathfrak{A}(M)\to\mathfrak{A}(N)
\ee
of the corresponding algebras of observables, moreover if $\chi_1:M\to N$ and $\chi_2:N\to L$ are embeddings as above then we require the covariance relation
\be
\alpha_{\chi_2\circ\chi_1}=\alpha_{\chi_2}\circ\alpha_{\chi_1} \ .
\ee
In this way we described a functor $\mathfrak{A}$ between two categories: the category $\Loc$ of globally hyperbolic spacetimes with admissible embeddings as arrows and the category $\Obs$ of algebras (Poisson algebras for classical physics and C*-algebras for quantum physics) with homomorphisms as arrows. 

We may restrict the category of spacetimes to subregions of a given spacetime and the arrows to inclusions. In this way we obtain the Haag-Kastler net of local algebras on a globally hyperbolic spacetime as introduced by Dimock. In case the spacetime has nontrivial isometries,  we obtain additional embeddings, and the covariance condition above provides a representation of the group of isometries by automorphisms of the Haag-Kastler net.

The causality requirements of the Haag-Kastler framework, i.e. the commutativity of observables localized in spacelike separated regions, is encoded in the general case in the tensor structure of the functor $\mathfrak{A}$. Namely, the category of globally hyperbolic manifolds has the disjoint union as a tensor product, with the empty set as unit object and where admissible embeddings $\chi:M_1\otimes M_2\to N$ have the property that the images $\chi(M_1)$ and $\chi(M_2)$ cannot be connected by a causal curve. On the level of C*-algebras we may use the minimal tensor product as a tensor structure. See \cite{BFIR} for details.

The solvability of the initial value problem can be formulated as the requirement that the algebra $\mathfrak{A}(N)$ of any neighborhood $N$ of some Cauchy surface $\Sigma$ already coincides with $\mathfrak{A}(M)$. This is the {\it time slice axiom} of axiomatic quantum field theory. It can be used to describe the evolution between different Cauchy surfaces. As a first step we associate to each Cauchy surface $\Sigma$ the inverse limit
\be
\mathfrak{A}(\Sigma)=\lim^{\leftarrow}_{N\supset \Sigma}\mathfrak{A}(N)\ .
\ee
Elements of the inverse limit consist of sequences $A=(A_N)_{L_A\supset N\supset \Sigma}$ with $\alpha_{N\subset K}(A_N)=A_K$, $K\subset L_A$, with the equivalence relation
\be
A\sim B \text{ if }A_N=B_N \text{ for all }N\subset L_A\cap L_{B}\ . 
\ee
The algebra $\mathfrak{A}(\Sigma)$ can be embedded into $\mathfrak{A}(M)$ by
\be
\alpha_{\Sigma\subset M}(A)=\alpha_{N\subset M}(A_N) \text{ for some (and hence all) }\Sigma\subset N\subset L_A\ .
\ee
If we now adopt the time slice axiom we find that each homomorphism $\alpha_{N\subset M}$ is an isomorphism. Hence $\alpha_{\Sigma\subset M}$ is also an isomorphism and we obtain the propagator between two Cauchy surfaces $\Sigma_1$ and $\Sigma_2$ by
\be
\alpha^M_{\Sigma_1\Sigma_2}=\alpha_{\Sigma_1\subset M}^{-1}\circ\alpha_{\Sigma_2\subset M}
\ee
This construction resembles constructions in topological field theory for the description of cobordisms. But there one associates Hilbert spaces to components of the boundary and maps between these Hilbert spaces to the spacetime itself. This construction relies on the fact that for these theories the corresponding Hilbert spaces are finite dimensional. It was shown \cite{Torr}
that a corresponding construction for the free field in 3 and more dimensions does not work, since the corresponding Boboliubov transformation is not unitarily implementable 
(Shale's criterion \cite{Sha} is violated). Instead one may associate to the Cauchy surfaces the corresponding algebras of canonical commutation relations and to the cobordism  an isomorphism between these algebras. For the algebra of canonical anticommutation relations for the free Dirac fields the above isomorphism was explicitly constructed \cite{Ver}. Our general argument shows that the association of a cobordism between two Cauchy surfaces of globally hyperbolic spacetimes to an isomorphism of algebras always exists provided the time slice axiom is satisfied. As recently shown, the latter axiom is actually generally valid in perturbative Algebraic Quantum Field Theory \cite{ChF}.     

In the Haag-Kastler framework on Minkowski space an essential ingredient was translation symmetry. This symmetry allowed the comparison of observables in different regions of spacetime and was  (besides locality) a crucial input for the analysis of scattering states.

In the general covariant framework sketched above no comparable structure is available. Instead one may use fields which are subject to a suitable covariance condition, termed locally covariant fields. A locally covariant field is a family $\ph_M$ of fields on spacetimes $M$ such that for every embedding $\chi:M\to N$ as above
\be
\alpha_{\chi}(\ph_M(x))=\ph_N(\chi(x)) \ .
\ee
If we consider fields as distributions with values in the algebras of observables, a field $\ph$ may be considered as a natural transformation between the functor $\mathfrak{D}$ of test function spaces to the functor $\mathfrak{A}$ of field theory.
The functor $\mathfrak{D}$ associates to every spacetime $M$ its space of compactly supported $\mathcal{C}^{\infty}$-functions,
\be
\mathfrak{D}(M)=\mathcal{C}^{\infty}_{\mathrm{c}}(M,\RR)\ ,
\ee
and to every embedding $\chi:M\to N$ of spacetimes the pushforward of test functions $f\in\mathfrak{D}(M)$
\be
\mathfrak{D}\chi\equiv \chi_*\ ,\ \chi_*f(x)=\left\{ \begin{array}{ccc}
                                                                                    f(\chi^{-1}(x)) &,&x\in\chi(M)\\
                                                                                    0                    &,&\text{ else}
                                                                                    \end{array}
                                                                           \right.
\ee
$\mathfrak{D}$ is a covariant functor. Its target category is the category of locally convex vector spaces $\Vect$ which contains also the category of topological algebras which is the target category for $\mathfrak{A}$. A natural transformation $\ph:\mathfrak{D}\to\mathfrak{A}$ between covariant functors with the same source and target categories is a family of morphisms $\ph_M:\mathfrak{D}(M)\to\mathfrak{A}(M)$, $M\in \mathrm{Obj}(\Loc)$ such that
\be
\mathfrak{A}\chi\circ\ph_M=\ph_N\circ\mathfrak{D}\chi
\ee
with $\mathfrak{A}\chi=\alpha_{\chi}$.
\section{Classical field theory}
Before we enter the arena of quantum field theory we show that the concept of local covariance leads to a nice reformulation of classical field theory
in which the relation to QFT becomes clearly visible. Let us consider a scalar field theory. On a given spacetime $M$ the possible field configurations are the smooth functions on $M$. If we embed a spacetime $M$ into another spacetime $N$, the field configurations on $N$ can be pulled back to 
$M$, and we obtain a functor $\mathfrak{E}$ from $\Loc$ to the category $\Vect$ of locally convex vector
spaces
\be
\mathfrak{E}(M)=\mathcal{C}^{\infty}(M,\RR)\ ,\ \mathfrak{E}\chi=\chi^*
\ee  
with the pullback $\chi^*\ph=\ph\circ\chi$ for $\ph\in\mathcal{C}^{\infty}(M,\RR)$. Note that $\mathfrak{E}$ is contravariant, whereas the functor $\mathfrak{D}$ of test function spaces with compact support is covariant.

The classical observables are real valued functions on $\mathfrak{E}(M)$, i.e. (not necessarily linear) functionals. An important property of a functional is its spacetime support. Is is defined as a generalization of the distributional support, namely as the set of points $x\in M$ such that $F$ depends on the field configuration in any neighbourhood of $x$.
\begin{align}\label{support}
\supp\, F\doteq\{ & x\in M|\forall \text{ neighbourhoods }U\text{ of }x\ \exists \ph,\psi\in\E(M), \supp\,\psi\subset U 
\\ & \text{ such that }F(\ph+\psi)\not= F(\ph)\}\ .\nonumber
\end{align}
Here we will discuss only compactly supported functionals. Next one has to select a class of functionals which are sufficiently regular such that all mathematical operations one wants to perform are meaningful and which on the other side is large enough to cover the interesting cases.

One class one may consider is the class $\F_\reg(M)$ of regular polynomials
\be
F(\ph)=\sum_{\text{finite}}\int dx_1\dots dx_n f_n(x_1,\dots,x_n)\ph(x_1)\dots\ph(x_n)
\ee
with test functions $f_n\in\mathfrak{D}(M^n)$. Another class $\F_\loc(M)$ consists of the local functionals
\be
F(\ph)=\int dx \Lcal(x,\ph(x),\pa\ph(x),\dots)
\ee
where $\Lcal$ depends smoothly on $x$ and on finitely many derivatives of $\ph$ at $x$. The local functionals arise as actions and induce the dynamics.
The only regular polynomials in this class are the linear functionals
\be
F(\ph)=\int dx f(x)\ph(x) \ .
\ee 
It turns out to be convenient to characterize the admissible class of functionals in terms of their functional derivatives.
\begin{df}[after \cite{Neeb}]
Let $X$ and $Y$ be topological vector spaces, $U \subseteq X$ an open set and $f:U \rightarrow Y$ a map. The derivative of $f$ at $x$ in the direction\index{derivative!on a locally convex vector space} of $h$ is defined as
\be\label{de}
df(x)(h) \doteq \lim_{t\rightarrow 0}\frac{1}{t}\left(f(x + th) - f(x)\right)
\ee
whenever the limit exists. The function $f$ is called differentiable\index{infinite dimensional!calculus} at $x$ if $df(x)(h)$ exists for all $h \in X$. It is called continuously differentiable if it is differentiable at all points of $U$ and
$df:U\times X\rightarrow Y, (x,h)\mapsto df(x)(h)$
is a continuous map. It is called a $\Ccal^1$-map if it is continuous and continuously differentiable. Higher derivatives are defined for $\Ccal^n$-maps by 
\be
d^n f (x)(h_1 , \ldots , h_n ) \doteq \lim_{t\rightarrow 0}\frac{1}{t}\big(d^{n-1} f (x + th_n )(h_1 , \ldots, h_{n-1} ) -
 d^{n-1}f (x)(h_1 , \ldots, h_{n-1}) \big)
 \ee
\end{df}
In particular it means that if $F$ is a smooth functional on $\E(M)$, then its $n$-th derivative at the point $\ph\in\E(M)$ is a compactly supported distributional density $F^{(n)}(\ph)\in\Ecal'(M^n)$. There is a distinguished volume form on $M$, namely the one provided by the metric: $\sqrt{-\textrm{det}(g)}d^4x$. We can use it to construct densities from functions and to provide an embedding of $\Dcal(M^n)$
 into $\Ecal'(M^n)$. For more details on distributions on manifolds, see chapter 1 of \cite{Baer}. Using the distinguished volume form we can identify derivatives $F^{(n)}(\ph)$ with distributions.
We further need some conditions on their wave front sets.

Let us make a brief excursion to the concept of wave front sets and its use for the treatment of distributions. Readers less familiar with these topics can find more details in the appendix \ref{distr}
or refer to \cite{Hoer} or chapter 4 of  \cite{BaeF}.
Let $t\in\mathcal{D}'(\RR^n)$ and $f\in\mathcal{D}(\RR^n)$. The Fourier transform of the product $ft$ is a smooth function. 
If this function vanishes fast at infinity for all $f\in\mathcal{D}(\RR^n)$, $t$ itself is a smooth function. Singularities of $t$ show up in the absence of fast decay in some directions. 
A point $(x,k)\in\RR^n\times(\RR^n\setminus\{0\})$ is called a regular point of $t$ if there exists a test function $f$ with $f(x)=1$ such that the Fourier transform of $ft$ decays strongly in an open cone around $k$. The wave front set of $t$ is now defined as the complement of the set of regular points  of $t$ in $\RR^n\times(\RR^n\setminus\{0\})$.

On a manifold $M$ the definition of the Fourier transform depends on the choice of a chart. But the property of strong decay in some direction (characterized now by a point $(x,k)$, $k\not=0$ of the cotangent bundle $T^*M$) turns out to be independent of the chart. Therefore the wave front set $\mathrm{WF}$ of a distribution on a manifold $M$ is a well defined closed conical subset of the cotangent bundle (with the zero section removed).

Let us illustrate the concept of the wave front set in two examples. The first one is the $\delta$-function. We find
\be
\int dx f(x)\delta(x)e^{ikx}=f(0) \ ,
\ee
hence $\mathrm{WF}(\delta)=\{(0,k),k\not=0\}$. 

The other one is the function $x\mapsto(x+i\epsilon)^{-1}$ in the limit $\epsilon\downarrow0$. We have
\be
\lim_{\epsilon\downarrow0}\int dx \frac{f(x)}{x+i\epsilon}e^{ikx}=-i\int_k^{\infty}dk'\hat{f}(k') \ .
\ee
Since the Fourier transform $\hat{f}$ of a test function $f\in\mathcal{D}(\RR)$ is strongly decaying for $k\to\infty$, $\int_k^{\infty}dk'\hat{f}(k')$ is strongly decaying for $k\to\infty$, but for $k\to -\infty$ we obtain
\be
\lim_{k\to-\infty}\int_k^{\infty}dk'\hat{f}(k')=2\pi f(0)\ , 
\ee  
hence
\be
\mathrm{WF}(\lim_{\epsilon\downarrow0}(x+i\epsilon)^{-1})=\{(0,k),k<0\}\ .
\ee

The wave front sets provide a simple criterion for the pointwise multiplicability of distributions. Namely, let $t,s$ be distributions on an $n$ dimensional manifold $M$ such that the pointwise sum (Whitney sum) of their wave front sets
\be
\mathrm{WF}(t)+\mathrm{WF}(s)=\{(x,k+k')|(x,k)\in\mathrm{WF}(t),(x,k')\in\mathrm{WF}(s)\}
\ee
does not intersect the zero section of $T^*M$. Then the pointwise product $ts$ can be defined by
\be
\langle ts,fg\rangle=\frac{1}{(2\pi)^n}\int dk\, \widehat{tf}(k)\widehat{sg}(-k)
\ee 
for test functions $f$ and $g$ with sufficiently small support and where the Fourier transform refers to an arbitrary chart covering the supports of $f$ and $g$. 
The convergence of the integral on the right hand side follows from the fact, that for every $k\not=0$ either $\widehat{tf}$ decays fast in a conical neighborhood around $k$ or 
$\widehat{sg}$ decays fast in a conical neighborhood around $-k$ whereas the other factor is polynomially bounded.

The other crucial property is the characterization of the propagation of singularities. To understand it in more physical terms it is useful to use an analogy with Hamiltonian mechanics. Note that the cotangent bundle $T^*M$ has a natural symplectic structure. The symplectic 2-form is defined as an exterior derivative of the canonical one-form, given in local coordinates as $
    \theta_{(x,{k})}=\sum_{{\mathfrak i}=1}^n {k}_idx^i $ (${k}_i$ are coordinates in the fibre).  Let $P$ be a partial differential operator with real principal symbol $\sigma_P$. Note that $\sigma_P$ is  a function on $T^*M$ and its differential $d\sigma_P$ is a 1-form. On a symplectic manifold 1-forms can be canonically identified with vector fields by means of the symplectic form. Therefore every differentiable function $H$ determines a unique vector field $X_H$, called the Hamiltonian vector field with the Hamiltonian $H$. Let $X_P$ be the Hamiltonian vector field corresponding to $\sigma_P$. Explicitly it can be written as:
\[
X_P=\sum_{i=1}^n\frac{\partial \sigma_P}{\partial {k}_j}\frac{\partial}{\partial x_j}-\frac{\partial \sigma_P}{\partial x_j}\frac{\partial}{\partial {k}_j}
\]
Let us now consider the integral curves (Hamiltonian flow) of this vector field. A curve $(x_j(t),{k}_j(t))$ is an integral curve of $X_P$ if it fulfills the system of equations (Hamilton's equations):
\begin{align*}
\frac{dx_j}{dt}&=\frac{\partial \sigma_P}{\partial {k}_j}\,,\\
\frac{d{k}_j}{dt}&=-\frac{\partial \sigma_P}{\partial x_j}\,.
\end{align*}
The set of all such solution curves is called the bicharacteristic flow. Along the Hamiltonian flow it holds
 $\frac{d\sigma_P}{dt}=X_P(\sigma_P)=0$ (this is the law of conservation of energy for autonomous systems in classical mechanics), so $\sigma_P$ is constant under the bicharacteristic flow. If $\sigma_P((x_j(t),{k}_j(t)))=0$ we call the corresponding flow null. The set of all such integral curves is called the null bicharacteristics.
 
Let us now define the characteristics of $P$ as $\mathrm{char}P=\{(x,k)\in T^*M|\sigma(P)(x,k)=0\}$ of $P$. Then the theorem on the propagation of singularities states that the  wave front set of a solution $u$ of the equation $Pu=f$ with $f$ smooth is a union of orbits of the Hamiltonian flow $X_P$ on the characteristics $\mathrm{char}P$. 

In field theory on Lorentzian spacetime we are mainly interested in hyperbolic differential operators. Their characteristics is the light cone, and the principal symbol is the metric on the cotangent bundle. The wave front set of solutions therefore is a union of null geodesics together with their cotangent vectors 
$k=g(\dot{\gamma},\cdot)$.

We already have  all the kinematical structures we need. Now in order to specify a concrete physical model we need to introduce the dynamics. This can be done by means of 
a \textit{generalized Lagrangian}\index{generalized Lagrangian}. As the name suggests the idea is motivated by Lagrangian mechanics. Indeed, we can think of this formalism as a way to make precise the variational calculus
in field theory. Note that since our spacetimes are globally hyperbolic, they are never compact. Moreover we cannot restrict ourselves to compactly supported field configurations, since the nontrivial solutions of globally hyperbolic equations don't belong to this class. Therefore we cannot identify the action with a functional on $\E(M)$ obtained 
by integrating the Lagrangian density over the whole manifold. Instead we follow \cite{BDF} and define a Lagrangian $L$ as a natural transformation between the functor of test function spaces $\D$ and the functor $\F_\loc$ such that it satisfies $\supp(L_M(f))\subseteq \supp(f)$ and the additivity rule 
\footnote{We do not require linearity since in quantum field theory the renormalization flow does not preserve the linear structure; it respects, however, the additivity rule (see \cite{BDF}).}
\[
L_M(f+g+h)=L_M(f+g)-L_M(g)+L_M(g+h)\,,
\]
for $f,g,h\in\D(M)$ and $\supp\,f\cap\supp\,h=\varnothing$.  
The action $S(L)$ is now defined as an equivalence class of Lagrangians  \cite{BDF}, where two Lagrangians $L_1,L_2$ are called equivalent $L_1\sim L_2$  if
\be\label{equ}
\supp (L_{1,M}-L_{2,M})(f)\subset\supp\, df\,, 
\ee
for all spacetimes $M$ and all $f\in\D(M)$. 
This equivalence relation allows us to identify Lagrangians differing by a total divergence.  For the free minimally coupled (i.e. $\xi=0$) scalar field the generalized Lagrangian is given by:
\be\label{Lscalar}
L_M(f)(\ph)=\frac{1}{2}\int\limits_M (\nabla_\mu\ph\nabla^\mu\ph-m^2\ph^2)f\,\dvol\,.
\ee

The equations of motion are to be understood in the sense of \cite{BDF}. Concretely, the Euler-Lagrange derivative\index{derivative!Euler-Lagrange} of $S$ is a natural transformation $S':\E\to\D'$ defined as
\be\label{ELd}
\left<S'_M(\ph),h\right>=\left<L_M(f)^{(1)}(\ph),h\right>\,,
 \ee
 with $f\equiv 1$ on $\supp h$. The field equation is now a condition on $\ph$:
\be
 S_M'(\ph)=0\,.\label{eom}
\ee
Note that the way we obtained the field equation is analogous to variational calculus on finite dimensional spaces. We can push this analogy even further and think of variation of a functional in a direction in configuration space given by an infinite dimensional vector field. This concept is well understood in mathematics and for details one can refer for example to \cite{Neeb,Michor}. Here we consider only variations in the directions of compactly supported configurations, so the space of vector fields we are interested in can be identified with
 $\mathfrak{V}(M)=\{X:\mathfrak{E}(M)\to\mathfrak{D}(M)| X\text{ smooth} \}$.
 In more precise terms this is the space of vector fields on $\mathfrak{E}(M)$, considered as a manifold\footnote{An infinite dimensional manifold is modeled on a locally convex vector space just as a finite dimensional one is modeled on $\RR^n$. For more details see \cite{Neeb,Michor}.} modeled over $\mathfrak{D}(M)$. The set of functionals
\be
\ph\mapsto\langle S_M'(\ph),X(\ph)\rangle\ , \  X\in\mathfrak{V}(M)
\ee
is an ideal $\mathfrak{I}_S(M)$ of $\mathfrak{F}(M)$ with respect to pointwise multiplication,
\be
(F\cdot G)(\ph)=F(\ph)G(\ph) \ .
\ee 
The quotient 
\be
\mathfrak{F}_S(M)=\mathfrak{F}(M)/\mathfrak{I}_S(M)
\ee
can be interpreted as the space of solutions of the field equation. The latter can be identified with the phase space of the classical field theory.

We now want to equip $\mathfrak{F}_S(M)$ with a Poisson bracket. Here we rely on a method originally introduced by Peierls. Peierls considers the influence of an additional term in the action. Let $F\in\mathfrak{F}_\loc(M)$ be a local functional. We are interested in the flow $(\Phi_{\lambda})$ on $\mathfrak{E}(M)$ which deforms solutions of the original field equation $S_M'(\ph)=\omega$ with a given source term $\omega$ to those of the perturbed equation $S_M'(\ph)+\lambda F^{(1)}(\ph)=\omega$. Let $\Phi_0(\ph)=\ph$ and 
\be\label{flow}
\frac{d}{d\lambda}\left.\left(S_M'(\Phi_{\lambda}(\ph))+F^{(1)}(\Phi_{\lambda}(\ph))\right)\right|_{\lambda=0}=0 \ .
\ee
Note that the second variational derivative of the unperturbed action induces an operator $S''_M(\ph):\E(M)\rightarrow\D'(M)$.  We define it in the following way:
\[
\left<S''_M(\ph),h_1\otimes h_2\right>\doteq \left<L^{(2)}_M(f)(\ph),h_1\otimes h_2\right>\,,
\]
where $f\equiv 1$ on the supports of $h_1$ and $h_2$. This defines $S''_M(\ph)$ as an element of $\D'(M^2)$ and by Schwartz's kernel theorem we can associate to it an operator from $\D(M)$ to $\D'(M)$.
Actually, since $L_M(f)$ is local, the second derivative has support on the diagonal, so $S''_M(\ph)$ 
can be evaluated on smooth functions $h_1$, $h_2$, where only one of them is required to be compactly supported, and it induces an operator (the so called linearized Euler-Lagrange operator) $E'[S_M](\ph):\E(M)\rightarrow \D'(M)$. 

From \eqref{flow} it follows that the vector field $\ph\mapsto X(\ph)=\frac{d}{d\lambda}\Phi_{\lambda}(\ph)|_{\lambda=0}$ satisfies the equation
\be\label{X:eq}
\left<S''_M(\ph),X(\ph)\otimes\cdot\right>+\left<F^{(1)}(\ph),\cdot\right>=0\,,
\ee
which in a different notation can be written as
\[
\left<E'[S_M](\ph),X(\ph)\right>+F^{(1)}(\ph)=0\,.
\]
We now assume that $E'[S_M](\ph)$ is, for all $\ph$, a normal hyperbolic differential operator $\E(M)\rightarrow\E(M)$, and let $\Delta_S^R,\Delta_S^A$ be the retarded and advanced Green's operators, i.e. linear operators  $\D(M)\rightarrow\E(M)$ satisfying:
\begin{align*}
E'[S_M]\circ\Delta^{R/A}_S&=\id_{\D(M)}\,,\\
\Delta^{R/A}_S\circ(E'[S_M]\big|_{\D(M)})&=\id_{\D(M)}\,.
\end{align*}
Moreover, with the use of Schwartz's kernel theorem one can identify $\Delta^{R/A}_S:\D(M)\rightarrow\E(M)$ with elements of $\Dcal'(M^2)$. As such, they are required to satisfy the following support properties:
\begin{align}
\supp(\Delta^R)&\subset\{(x,y)\in M^2| y\in J^-(x)\}\label{retarded}\,,\\
\supp(\Delta^A)&\subset\{(x,y)\in M^2| y\in J^+(x)\}\label{advanced}\,.
\end{align} 
Their difference $\De_S= \De_S^{\rm A} -\De_S^{\rm R}$ 
is called the causal propagator of the Klein-Gordon equation. 
Coming back to equation \eqref{X:eq}
we have now two distinguished solutions for  $X$,
\be
X^{R,A}(\ph)=\Delta_S^{R,A}F^{(1)}(\ph)\ .
\ee
The difference of the associated derivations on $\mathfrak{F}(M)$ defines a product
\be\label{Peierls:bracket}
\Poi{F}{G}_S(\ph)=\langle \Delta_S(\ph) F^{(1)}(\ph),G^{(1)}(\ph)\rangle
\ee
on $\mathfrak{F}_\loc(M)$, the so-called Peierls bracket.

The Peierls bracket satisfies the conditions of a Poisson bracket, in particular the Jacobi identity (for a simple proof see \cite{Jak}). Moreover, if one of the entries is in the ideal $\mathfrak{I}_S(M)$, also the bracket is in the ideal, hence the Peierls bracket induces a Poisson bracket on the quotient algebra.

In standard cases, the Peierls bracket coincides with the canonical Poisson bracket. Namely let
\be
\mathcal{L}(\ph)=\frac12\pa^\mu\ph\pa_\mu\ph-\frac{m^2}{2}\ph^2-\frac{\lambda}{4!}\ph^4\ .
\ee 
Then $S'_M(\ph)=-\left((\square+m^2)\ph+\frac{\lambda}{3!}\ph^3\right)$ and $S_M''(\ph)$ is the linear operator
\be
-\left(\square +m^2+\frac{\lambda}{2}\ph^2\right)
\ee
(the last term acts as a multiplication operator).

The Peierls bracket is
\be
\Poi{\ph(x)}{\ph(y)}_S=\Delta_S(\ph)(x,y)
\ee
where $x\mapsto \Delta_S(\ph)(x,y)$ is a solution of the (at $\ph$) linearized equation of motion with the initial conditions
\be
\Delta_S(\ph)(y^0,\mathbf{x};y)=0\ ,\ \frac{\pa}{\pa x^0}\Delta_S(\ph)(y^0,\mathbf{x};y)=\delta(\mathbf{x},\mathbf{y})\ .
\ee
This coincides with the Poisson bracket in the canonical formalism. Namely, let $\ph$ be a solution of the field equation.
Then
\be
0=\{(\square +m^2)\ph(x)+\frac{\lambda}{3!}\ph^3(x),\ph(y)\}=(\square +m^2+\frac{\lambda}{2}\ph(x)^2))\{\ph(x),\ph(y)\}
\ee
hence the Poisson bracket satisfies the linearized field equation with the same initial conditions as the Peierls bracket.

Let us now discuss the domain of definition of the Peierls bracket. It turns out that it is a larger class of functionals than just $\F_\loc(M)$. To identify this class we use the fact that the WF set of $\Delta_S$ is given by
\[
\WF(\Delta_S) = \{(x,k;x',-k') \in \dot{T}^*M^2 | (x,k) \sim (x',k')\}\,,
\]
where the equivalence relation $\sim$ means that there exists a null geodesic strip such that both $(x,k)$ and $(x',k')$ belong to it. A null geodesic strip is a curve in $T^*M$ of the form $(\gamma(\lambda),k(\lambda))$, $\lambda\in I\subset \RR$, where $\gamma(\lambda)$ is a null geodesic parametrized by $\lambda$ and $k(\lambda)$ is given by $k(\la)=g(\dot{\gamma}(\la),\cdot)$. This follows from the theorem on the propagation of singularities together with the initial conditions and the antisymmetry of $\Delta_S$. (See \cite{Rad} for a detailed argument.)

It is now easy to check, using H\"ormander's criterion on the multiplicability of distributions \cite{Hoer} that the Peierls bracket \eqref{Peierls:bracket} is well defined if   $F$ and $G$ are such that the sum of the WF sets of the functional derivatives $F^{(1)}(\ph), G^{(1)}(\ph)\in\Ecal'(M)$ and $\Delta\in\Dcal'(M^2)$ don't intersect the 0-section of the cotangent bundle $T^*M^2$. 
This is the case if the functionals fulfill the following criterion: 
 \be\label{mlsc}
\WF(F^{(n)}(\ph))\subset \Xi_n,\quad\forall n\in\NN,\ \forall\ph\in\E(M)\,,
\ee
where $\Xi_n$ is an open cone defined as 
\be\label{cone}
\Xi_n\doteq T^*M^n\setminus\{(x_1,\dots,x_n;k_1,\dots,k_n)| (k_1,\dots,k_n)\in (\overline{V}_+^n \cup \overline{V}_-^n)_{(x_1,\dots,x_n)}\}\,,
\ee
where $(\overline{V}_{\pm})_x$ is the closed future/past lightcone understood as a conic subset of
$T^*_xM$. We denote the space of smooth compactly supported functionals, satisfying (\ref{mlsc}) by $\F_\mc(M)$ and call them \textit{microcausal functionals}. This includes in particular local functionals. For them the support of the functional derivatives is on the thin diagonal, and the wave front sets satisfy $\sum k_i=0$.

To see that $\Poi{.}{.}_S$ is indeed well defined on  $\F_\mc(M)$, note that $\WF(\Delta)$ consists of elements $(x,x',k,k')$, where $k$, $k'$ are dual to lightlike vectors in $T_xM$, $T_{x'}M$ accordingly. On the other hand, if $(x,k_1)\in\WF(F^{(1)}(\ph))$, then $k_1$ is necessarily dual to a vector which is spacelike, so $k_1+k$ cannot be 0. The same argument is valid for $G^{(1)}(\ph)$. Moreover it can be shown that $\Poi{F}{G}_S\in\F_{\mc}(M)$. 
The classical field theory is defined as $\fA(M)=(\F_\mc(M),\Poi{.}{.}_S)$. One can check that $\fA$ is indeed a covariant functor from $\Loc$ to $\Obs$, the category of Poisson algebras.
\section{Deformation quantization of free field theories}\label{deformation}
Starting from the Poisson algebra $(\mathfrak{F}_{\mc}(M),\Poi{.}{.}_S)$ one may try to construct an associative algebra $(\mathfrak{F}_\mc(M)[[\hbar]],\star)$ such that for $\hbar\to0$
\be
F\star G\to F\cdot G 
\ee 
and
\be
[F,G]_{\star}/i\hbar\to \Poi{F}{G}_S \ .
\ee
For the Poisson algebra of functions on a finite dimensional Poisson manifold the deformation quantization exists in the sense of formal power series due to a theorem of Kontsevich \cite{Kontsevich:1997vb}. 
In field theory the formulas of Kontsevich lead to ill defined terms, and a general solution of the problem is not known. But in case the action is quadratic in the fields, the $\star$-product can be explicitly defined by
\be
(F\star G)(\ph)\doteq\sum\limits_{n=0}^\infty \frac{\hbar^n}{n!}\left<F^{(n)}(\ph),\left(\tfrac{i}{2}\Delta_S\right)^{\otimes n}G^{(n)}(\ph)\right>\,,
\ee 
which can be formally written as $e^{\frac{i\hbar}{2}\left\langle\Delta_S,\frac{\delta^2}{\delta\ph\delta\ph'}\right\rangle}F(\ph)G(\ph')|_{\ph'=\ph}$.
This product is well defined (in the sense of formal power series in $\hbar$) for regular functionals $F,G\in\F_\reg(M)$ and satisfies the conditions above.
Let for instance
\be
F(\ph)=e^{\int dx\ph(x)f(x)}\ ,\ G(\ph)=e^{\int dx\ph(x)g(x)}\ ,
\ee
with test functions $f,g\in\mathfrak{D}(M)$. We have 
\be
\frac{\delta^n}{\ph(x_1)\dots\ph(x_n)}F(\ph)=f(x_1)\dots f(x_n)F(\ph)
\ee
and hence
\begin{gather}
(F\star G)(\ph)\\
=\sum_{n=0}^{\infty}\frac{1}{n!}\left(\int dxdy\frac{i\hbar}{2}\Delta_S(x,y)f(x)g(y)\right)^n\, F(\ph)G(\ph)
\end{gather}

For later purposes we want to extend the product to more singular functionals which includes in particular the local functionals. We split
\be
\tfrac{i}{2}\Delta_S=\Delta_S^{+}-\Delta_S^{-}
\ee
in such a way that the wavefront set of $\Delta_S$ is decomposed into two disjoint parts. The wave front set of $\Delta_S$ consists of pairs of points $x,x'$ which can be connected by a null geodesic, and of covectors $(k,k')$ where $k$ is the cotangent vector of the null geodesic at $x$ and $-k'$ is the cotangent vector of the same null geodesic at $x'$. The lightcone with the origin removed consists of two disjoint components,  the first one containing the positive frequencies and the other one the negative frequencies. The WF set of the positive frequency part of $\Delta_S$ is therefore:
\begin{equation}\label{spectrum}
\WF(\Delta_S^+)=\{(x,k;x′,-k')\in \dot{T}M^2|(x,k)\sim(x',k'), k\in (\overline{V}_+)_x\}\,.
\end{equation}
 On Minkowski space one could choose $\Delta_S^+$ as the  Wightman 2-point-function, i.e. the vacuum expectation value of the product of two fields. This, however, becomes meaningless in a more general context, since a generally covariant concept of a vacuum state does not exist. Nevertheless, such a decomposition always exist, but is not unique and the difference between two different choices of $\Delta_S^+$ is always a smooth symmetric function. Let us write $\Delta_S^+=\tfrac{i}{2}\Delta_S+H$
We then consider the linear functional derivative operator
\be
\Gamma_H=\langle H,\frac{\delta^2}{\delta\ph^2}\rangle
\ee
and define a new $\star$-product by
\be
F\star_HG=\al_H\left((\al^{-1}_HF)\star(\al_H^{-1}G)\right)\,, 
 \ee
 where $\al_H\doteq e^{\frac{\hbar}{2}\Gamma_H}$.  This product differs from the original one in the replacement of $\frac{i}{2}\Delta_S$ by $\Delta_S^+$. 
This star product can now be defined on a much larger space of functionals, namely the microcausal ones $\F_\mc(M)$. The transition between these two star products correspond to normal ordering, and the $\star_H$-product is just an algebraic version of Wick's theorem. The map $\al_H$ provides the equivalence between $\star$ and $\star_H$ on the space of regular functionals $\F_\reg(M)$. Its image can be then completed to a larger space $\F_\mc(M)$. We can also build a corresponding (sequential) completion $\al_H^{-1}(\F_\mc(M))$ of the source space. This amounts to extending  $\F_\reg(M)$ with all elements of the form $\lim_{n\rightarrow \infty}\al_H^{-1}(F_n)$, where $(F_n)$ is a convergent sequence in $\F_\mc(M)$ with respect to the H\"ormander topology \cite{BDF,Hoer}. We recall now the definition of this topology. 
\begin{df}
Let us denote the space of compactly supported distributions with WF sets contained in a conical set $C\subset T^*M^n$ by   $\Ecal'_{C}(M^n)$. 
Now let  $C_n\subset \Xi_n$ be a closed cone contained in $\Xi_n$ defined by \eqref{cone}. We introduce (after \cite{Hoer,BaeF,BDF}) the following family of seminorms on $\Ecal'_{C_n}(M^n)$: 
\[
p_{n,\ph,\tilde{C},k} (u) = \sup_{{k}\in \tilde{C}}\{(1 + |{k}|)^k |\widehat{\ph u}({k})|\}\,,
\]
where the index set consists of $(n,\ph,\tilde{C},k)$ such that $k\in \NN_0$, $\ph\in \Dcal(M)$ and $\tilde{C}$ is a closed cone in $\RR^n$ with $(\supp ( \ph ) \times \tilde{C}) \cap C_n = \varnothing$. These seminorms, together with the seminorms of the weak topology provide a defining system for a locally convex topology denoted by $\tau_{C_n}$. To control the wave front set properties inside open cones, we take an inductive limit.  The resulting topology is denoted by $\tau_{\Xi_n}$. One can show that $\Dcal(M)$ is sequentially dense in  $\Ecal'_{\Xi_n}(M)$ in this topology.

For microcausal functionals it holds that $F^{(n)}(\ph)\in\Ecal'_{\Xi_n}(M)$, so we can equip $\F_\mc(M)$ with the initial topology with respect to mappings:
\be\label{tauH}
\Ci(\E(M),\RR)\ni F\mapsto F^{(n)}(\ph)\in(\Ecal_{\Xi_n}(M),\tau_{\Xi_n})\quad n\geq0\,,
\ee
and denote this topology by $\tau_\Xi$.
\end{df}
The locally convex vector space of local functionals $\F_\loc(M)$ is dense in $\F_\mc(M)$ with respect to $\tau_\Xi$. To see these abstract concepts at work let us consider the example of the Wick square:
 \begin{exa}\label{Wick}
Consider a sequence $F_n(\ph)=\int \ph(x)\ph(y)g_n(y-x)f(x)$ with a smooth function $f$ and a sequence of smooth functions $g_n$ which converges to the $\delta$ distribution in the H\"ormander topology. By applying $\al_H^{-1}=e^{-\tfrac{\hbar}{2}\Gamma_H}$, we obtain a sequence 
\[
\al_H^{-1}F_n= \int (\ph(x)\ph(y)g_n(y-x)f(x)-H(x,y)g_n(y-x)f(x))\,,
\]
The limit of this sequence can be identified with $\int :\ph(x)^2:f(x)$, i.e.:
\[
\int :\ph(x)^2:f(x)=\lim_{n\rightarrow\infty}\int (\ph(x)\ph(y)-H(x,y))g_n(y-x)f(x)
\]
We can write it in a short-hand notation as a coinciding point limit:
\[
:\ph(x)^2:\,=\lim_{y\rightarrow x}(\ph(x)\ph(y)-H(x,y))\,.
\]
 We can see that transforming with $\al_H^{-1}$ corresponds formally to a subtraction of $H(x,y)$.
Now, to recognize the Wick's theorem let us consider a product of two Wick squares $ :\ph(x)^2: :\ph(y)^2:$. With the use of the isomorphism $\al_H^{-1}$ this can be written as:
\begin{multline*}
\int \ph(x)^2 f_1(x)\star_H\int \ph(y)^2 f_2(y)=\\
\int \ph(x)^2 \ph(y)^2f_1(x) f_2(y)+2i\hbar \int \ph(x) \ph(y)\De_S^+(x,y)f_1(x) f_2(y)-\hbar^2\int (\De_S^+(x,y))^2f_1(x) f_2(y)\,.
\end{multline*}
Omitting the test functions and using $\al_H^{-1}$ we obtain
\[
:\ph(x)^2::\ph(y)^2:=:\ph(x)^2 \ph(y)^2:+4 :\ph(x) \ph(y):\frac{i\hbar}{2}\De_S^+(x,y)+2\Big(\frac{i\hbar}{2}\De_S^+(x,y)\Big)^2\,,
\]
which is a familiar form of the Wick's theorem applied to $:\ph(x)^2::\ph(y)^2:$.
\end{exa}
Different choices of $H$ differ only by a smooth function, hence all the algebras  $(\al_H^{-1}(\F_\mc(M)[[\hbar]]),\star)$ are isomorphic and define an abstract algebra $\fA(M)$. Since for $F \in \fA(M)$ it holds $\al_HF\in \fA_H(M)=(\F_{\mc}(M)[[\hbar]],\star_H)$, we can realize $\fA(M)$ more concretely as the space of families $\{ \al_HF \}_H$,  numbered by possible choices of $H$, where  $F \in \fA(M)$, fulfilling the relation
\[
 F_{H'} = \exp(\hbar \Gamma_{H'-H}) F_H\,,
\]
equipped with the product
\[
 (F \star G)_H = F_H \star_H G_H.
\]
The support of $F \in \fA(M)$ is defined as $\supp(F) = \supp(\al_HF)$. Again, this is indepedent of $H$. Functional derivatives are defined by
\[
\left<\frac{\delta F}{\delta \ph},\psi\right> = \al_H^{-1}\left<\frac{\delta \al_HF}{\delta \ph},\psi\right>\,,
\]
which is well defined as $\Gamma_{H'-H}$ commutes with functional derivatives. In the next step we want to define the involution on our algebra. 
Note that the complex conjugation satisfies the relation:
\be
\overline{F\star G}=\overline{G}\star\overline{F}\,.
\ee
Therefore we can use it to define an involution  $F^*(\ph)\doteq\overline{F(\ph)}$. The resulting structure is an involutive noncommutative algebra $\fA(M)$, which provides a quantization of $(\F_{\mc}(M),\Poi{.}{.}_S)$. To see that this is equivalent to canonical quantization, let us look at the commutator of two smeared fields  $\Phi(f)$,  $\Phi(g)$, where $\Phi(f)(\ph)\doteq \int\! f \ph\,\  \dvol $. The commutator reads
\[
[\Phi(f),\Phi(g)]_{\star_H}=i\hbar \langle f,\De_S g\rangle \ , \quad f,g\in \D(M)\,,
\]
This indeed reproduces the canonical commutation relations. Here we used the fact that the choice of $\Delta_S^+$ is unique up to a symmetric function, which doesn't contribute to the commutator (which is antisymmetric). In case $\Delta_S^+$ is a distribution of positive type (as in the case of the Wightman 2-point-function) the linear functional on $\mathfrak{F}(M)$
 \be
 \om(F)=F(0)
 \ee 
 is a state (the vacuum state in the special case above), and the associated GNS representation is the Fock representation. The kernel of the representation is the ideal generated by the field equation.
 
 Let us now discuss the covariance properties of Wick products. As seen in example \ref{Wick}, polynomial functionals in $\fA_H(\Mcal)$ can be interpreted as Wick polynomials.
Corresponding elements of $\fA(\Mcal)$ can be obtained by applying $\al_H^{-1}$. The resulting object will be denoted by
\be\label{polynomials1}
\int :\Phi_{x_1}\dots\Phi_{x_n}:_H f(x_1,\dots, x_n)\doteq \al^{-1}_H\Big(\int \Phi_{x_1}\dots\Phi_{x_n} f(x_1,\dots, x_n)\Big)\,.
\ee
where $\Phi_{x_i}$ are evaluation functionals, $f\in\Ecal'_{\Xi_n}(M^n,V)$. 
 
 The assignment of $\fA(M)$ to a spacetime $M$ can be made into a functor $\fA$ from the category $\Loc$ of spacetimes to the category of  topological *-algebras $\Obs$ and by composing with a forgetful functor to the category $\Vect$ of topological vector spaces. Admissible embeddings are mapped to pullbacks, i.e. for $\chi:M\rightarrow M'$ we set $\fA\chi F(\ph)\doteq F(\chi^*\ph)$. Locally covariant quantum fields are natural transformations between $\D$ and $\fA$.  Let us denote the extended space of locally covariant quantum fields by $\Fcal_q$. We shall require Wick powers to be elements of $\Fcal_q$ in the above sense. On each object $M$ we have to construct the map ${\TT_1}_M$ from the classical algebra $\F_\loc(M)$ to the quantum algebra $\fA(M)$ in such a way that
\be
\label{covariance}
{\TT_1}_{M}(\Phi_{M}(f))(\chi^*\ph)={\TT_1}_{M'}(\Phi_{M'}(\chi_*f))(\ph)\,.
\ee
As we noted above, classical functionals can be mapped  to $\fA_H(M)$ by identification \eqref{polynomials1}. This however doesn't have right covariance properties. A detailed discussion is presented in section 5 of \cite{BFV}. Here we only give a a sketch of the argument for the Wick square. For each object $M\in\Loc$ we choose $H_M$  (so ${\TT_1}_M=\al^{-1}_{H_M}$ )  and  going through the definitions it is easy to see that, for an admissible embedding $\chi:M\rightarrow M'$
\[
\fA\chi\big(:\Phi^2:_{H_M}(x)\big)=:\Phi^2:_{H_{M'}}(\chi(x))+H_{M'}(\chi(x),(x))-H_{M}(x,x)
\]
holds. It was shown in \cite{BFV} that redefining Wick powers to become covariant amounts to solving certain cohomological problem. The result reproduces the solution, proposed earlier in \cite{HW}, to define $\TT_1$ as $\al^{-1}_{H+w}$, where $w$ is the smooth part of the Hadamard 2-point function $\omega=\frac{u}{\sigma}+v\ln\sigma+w$ with $\sigma(x,y)$ denoting the square of the length of the geodesic connecting $x$ and $y$ and with geometrical determined smooth functions $u$ and $v$.
\section{Interacting theories and the time ordered product}
If we have an action for which $S''_M$ still depends on $\ph$, we choose a particular $\ph_0$ and split
\be
S_M(\ph_0+\psi)=\frac12 \langle S_M''(\ph_0),\psi\otimes\psi\rangle +S_I(\ph_0,\psi) \ .
\ee
From now on we drop the subscript $M$ of $S_I$, since it's clear that we work on a fixed manifold.
We now introduce the linear operator
\be
\TT=e^{i\hbar\langle\Delta_S^D,\frac{\delta^2}{\delta\psi^2}\rangle} 
\ee
which acts on $\F_\reg(M)$ as
\[
(\TT F)(\ph)\doteq \sum_{n=0}^\infty \frac{\hbar^n}{n!}\left<(i\Delta_S^D)^{\otimes n},F^{(2n)}(\ph)\right>,\,,
\]
with the Dirac propagator $\Delta_S^D=\frac12(\Delta_S^R+\Delta_S^A)$ at $\ph_0$. Formally, $\TT$ may be understood as the operator of convolution with the oscillating Gaussian measure with covariance $i\hbar\Delta_S^D$. By
\be
F\T G=\TT\left(\TT^{-1}F\cdot \TT^{-1}G\right)
\ee
we define a new product on $\mathfrak{F}(M)$ which is the time ordered product with respect to $\star$ and which is equivalent to the pointwise product of classical field theory.
We then define a linear map
\be
R_{S_I}F=\Big(e_{\sst{\TT}}^{S_I}\Big)^{\star-1}\star \left(e_{\sst{\TT}}^{S_I}\T F\right)
\ee
where $e_{\sst{\TT}}$ is the exponential function with respect to the time ordered product,
\be
e_{\sst{\TT}}^F=\TT\big( e^{\TT^{-1}F}\big) \ .
\ee
$R_{S_I}$ is invertible with the inverse
\be
R_{S_I}^{-1}F=e_T^{-S_I}\cdot_T\left(e_T^{S_I}\star F\right)
\ee
We now define the $\star$ product for the full action by
\be
F\star_SG=R_{S_I}^{-1}\left(R_{S_I}F\star R_{S_I}G\right)
\ee
\global\long\def\poisson#1#2{\left\lfloor #1,#2\right\rceil }

\global\long\def\bld#1{\boldsymbol{#1}}

\begin{fmffile}{SettingPAQFT}

\def\FD{\parbox{7mm}{
\begin{center}
\begin{fmfgraph}(5,5)
\fmfleft{F}
\fmfdot{F}
\end{fmfgraph}
\end{center}}}

\def\FFlop{\parbox{20mm}{
\begin{center}
\begin{fmfgraph}(20,15)
\fmfleft{F} \fmfright{G}
\fmfdot{F}
\fmf{phantom}{F,G}
\fmf{plain}{F,F}
\end{fmfgraph}
\end{center}}}

\def\FG{\parbox{13mm}{
\begin{center}
\begin{fmfgraph}(20,15)
\fmfleft{F} \fmfright{G}
\fmfdot{F,G}
\fmf{phantom}{F,G}
\end{fmfgraph}
\end{center}}}

\def\dumbbell{\parbox{10mm}{
\begin{center}
\begin{fmfgraph}(20,10)
\fmfleft{F} \fmfright{G}
\fmfv{decor.shape=circle,decor.filled=empty, decor.size=2thick}{F,G}
\fmf{plain}{F,G}
\end{fmfgraph}
\end{center}}}

\def\FoneFG{\parbox{20mm}{
\begin{center}
\begin{fmfgraph}(20,15)
\fmfleft{F} \fmfright{G}
\fmfdot{F,G}
\fmf{phantom}{F,G}
\fmf{plain}{F,F}
\end{fmfgraph}
\end{center}}}

\def\FGoneG{\parbox{20mm}{
\begin{center}
\begin{fmfgraph}(20,15)
\fmfleft{F} \fmfright{G}
\fmfdot{F,G}
\fmf{phantom}{F,G}
\fmf{plain}{G,G}
\end{fmfgraph}
\end{center}}}

\def\FoneG{\parbox{13mm}{
\begin{center}
\begin{fmfgraph}(20,15)
\fmfleft{F} \fmfright{G}
\fmfdot{F,G}
\fmf{plain}{F,G}
\end{fmfgraph}
\end{center}}}

\def\FdecoG{\parbox{15mm}{
\begin{center}
\begin{fmfgraph*}(20,15)
\fmfleft{F} \fmfright{G}
\fmflabel{$F$}{F}
\fmflabel{$G$}{G}
\fmfdot{F,G}
\fmf{fermion}{F,G}
\end{fmfgraph*}
\end{center}}}
\section{Renormalization}
Unfortunately, the algebraic structures discussed so far are well defined only if $S_I$ is a regular functional. An easy extension is provided by the operation of normal ordering as described in the section of deformation quantization. This operation transforms the time ordering operator $\TT$ into another one $\TT'$, such that the new time ordered product is now defined with respect to the Feynman propagator $\Delta_S^F=i\Delta_S^D+H$, no longer the Dirac propagator $\Delta_S^D$. Note that the Feynman propagator does depend on the choice of $\Delta_S^+$.  Contrary to the $\star_H$ product which is everywhere defined due to the wave front set properties of the positive frequency part of $\Delta_S$ the time ordered product is in general undefined since the wave front set of the Feynman propagator contains the wave front set of the $\delta$-function.  We want, however, to extend to a larger class which contains in particular
all local functionals. As already proposed by St\"uckelberg \cite{Stuck} and Bogoliubov \cite{BP,BS} and carefully worked out by Epstein and Glaser \cite{EG}, the crucial problem is the definition of time ordered products of local functionals.
Let us first consider a special case.

Let $F=\frac12\int dx \ph(x)^2f(x)$, $G=\frac12\int dx \ph(x)^2g(x)$. Then
the time ordered product $\cdot_{\TT'}$ is formally given by
\be\label{ex:DeltaF}
(F\cdot_{\TT'}G)(\ph)=F(\ph)G(\ph)+i\hbar \int dxdy\ph(x)\ph(y)f(x)g(y)\Delta_S^F(x,y)-\frac{\hbar^2}{2}\int dxdy \Delta_S^F(x,y)^2 f(x)g(y)\ .
\ee 
But the last term contains the pointwise product of a distribution with itself. For $x\not=y$ the covectors $(k,k')$ in the wave front set satisfy the condition that $k$ and $-k'$ are cotangent to an (affinely parametrized) null geodesics connecting 
$x$ and $y$. $k$ is future directed if $x$ is in the future of $y$ and past directed otherwise. Hence the sum of two such covectors cannot vanish. Therefore the theorem on the multiplicability of distributions applies and yields a distribution on the complement of the diagonal $\{(x,x)|x\in M\}$.
On the diagonal, however, the only restriction is $k=-k'$, hence the sum of the wave front set of $\Delta_S^F$ with itself meets the zero section of the cotangent bundle at the diagonal. 

In general the time-ordered product $\TT_{n}(F_1,\ldots,F_n)\doteq F_1\T\dots\T F_n$ of $n$ local functionals is well defined for local entries as long as supports of $F_1,\ldots,F_n$ are pairwise disjoint. 
The technical problem one now has to solve is the extension of a distribution which is defined outside of a submanifold to an everywhere defined distribution. In the case of QFT on Minkowski space one can exploit translation invariance and reduce the problem in the relative coordinates to the extension problem of a distribution defined outside of the origin in $\RR^n$. The crucial concept for this extension problem is Steinmann's scaling degree \cite{Steinmann}.

\begin{df}
Let $U\subset \RR^n$ be a scale invariant open subset (i.e. $\lambda U=U$ for $\lambda>0$), and let $t\in\Dcal'(U)$ be a distribution on $U$. Let
$t_{\la}(x)=t(\la x)$ be the scaled distribution. The scaling degree $\mathrm{sd}$ of $t$ is 
\be
\mathrm{sd}\,t=\mathrm{inf}\{\delta\in\RR|\lim_{\la\to0}\la^{\delta}t_{\la}=0\} \ .
\ee  
\end{df} 
There is one more important concept related to the scaling degree, namely the degree of divergence. It is defined as:
\[
\mathrm{div}(t)\doteq \mathrm{sd}(t)-n\,.
\]
\begin{thm}\label{extension}
Let $t\in \Dcal(\RR^n\setminus\{0\})$ with scaling degree $\mathrm{sd}\,t<\infty$. Then there exists an extension of $t$ to an everywhere defined distribution with the same scaling degree. The extension is unique up to the addition of a derivative $P(\partial)\delta$ of the delta function, where $P$ is a polynomial with degree bounded by 
$\mathrm{div}(t)$ (hence vanishes for $\mathrm{sd}\,t<n$).  
\end{thm}
A proof may be found in \cite{DF04}. In the example above the scaling degree of $\Delta_S^F(x)^2$ is 4 (in 4 dimensions). Hence the extension exists and is unique up to the addition of a multiple of the delta function.

The theorem above replaces the cumbersome estimates on conditional convergence of Feynman integrals on Minkowski momentum space. Often this convergence is not proven at all, instead the convergence of the corresponding integrals on momentum space with euclidean signature is shown. The transition to Minkowski signature is then made after the integration. This amounts not to a computation but merely to a definition of the originally undefined Minkowski space integral.

The generalization of the theorem on the extension of distributions to the situation met on curved spacetimes is due to Brunetti and one of us (K.F.) \cite{BF0}. It uses techniques of microlocal analysis to reduce the general situation to the case covered by the theorem above.

The construction of time ordered products is then performed in the following way ({\it causal perturbation theory}). One searches for a family $(\TT_n)_{n\in\NN_0}$ of $n$-linear symmetric maps from local functionals to microcausal functionals subject to the following conditions:
\begin{enumerate}[{\bf T 1.}]
\item $\TT_0=1$,
\item $\TT_1=\mathrm{id}$ (for curved spacetime one should rather choose $\TT_1=e^{\Gamma_w}$, see the discussion at the end of the section \ref{deformation}),
\item $\TT_n(F_1,\dots,F_n)=\TT_k(F_1,\dots,F_k)\star \TT_{n-k}(F_{k+1},\dots,F_n)$ if the supports $\supp F_i$, $i=1,\dots,k$ of the first $k$ entries do not intersect the past of the supports $\supp F_j$, $j=k+1,\dots,n$ of the last $n-k$ entries (\textit{causal factorisation property}).\label{causality}
\end{enumerate} 
The construction proceeds by induction: when the first $n$ maps $\TT_k$, $k=0,\dots,n $ have been determined, the map $\TT_{n+1}$ is determined up to an $(n+1)$-linear map $Z_{n+1}$ from local functionals to local functionals. This ambiguity corresponds directly to the freedom of adding finite counterterms in every order in perturbation theory.

The general result can be conveniently formulated in terms of the formal S-matrix, defined as the generating function of time ordered products,
\be
\Scal(V)=\sum_{n=0}^{\infty}\frac{1}{n!}\TT_n(V,\dots,V) \ .
\ee 
Then the S-matrix $\hat{\Scal}$ with respect to an other sequence of time ordered products is related to $\Scal$ by
\be
\hat{\Scal}=\Scal\circ Z
\ee
where $Z$ maps local functionals to local functionals, is analytic with vanishing zero order term and with the first order term being the identity. The maps $Z$ form the renormalization group in the sense of Petermann and St\"uckelberg. They are formal diffeomorphisms on the space of local functionals and describe the allowed finite renormalization.

In order to illustrate the methods described above we work out the combinatorics in terms of Feynman diagrams (graphs). 
Let $D$ be the second order functional differential operator $D=\frac{i\hbar}{2}\langle\Delta_S^F,\frac{\delta^2}{\delta\ph^2}\rangle$.
The time ordered product of $n$ factors is formally given by
\[F_1\cdot_\TT\dots F_n\equiv\TT_n(F_1,\dots F_n)=e^{\frac12D}(e^{-\frac12 D}F_1\cdot\dots e^{-\frac12D}F_n)\]
Using Leibniz' rule and the fact that $D$ is of second order we find
\be\label{expD}
(F_1\cdot_\TT\dots F_n)(\ph)=e^{\sum_{i<j}D_{ij}}F_1(\ph_1)\cdots F_n(\ph_n)|_{\ph_1=\dots\ph_n=\ph}
\ee
with $D_{ij}=i\hbar\langle\Delta_S^F,\frac{\delta^2}{\delta\ph_i\delta\ph_j}\rangle$.
The expansion of the exponential function of the differential operator yields
\be\label{expDij}
e^{\sum_{i<j}D_{ij}}=\prod_{i<j}\sum_{l_{ij}=0}^{\infty}\frac{D_{ij}^{l_{ij}}}{l_{ij}!}
\ee
The right hand side may now be written as a sum over all graphs $\Gamma$ with vertices $V(\Gamma)=\{1,\dots,n\}$ and $l_{ij}$ lines $e\in E(\Gamma)$ connecting the vertices $i$ and $j$.
We set $l_{ij}=l_{ji}$ for $i>j$ and $l_{ii}=0$ (no tadpoles). If $e$ connects $i$ and $j$ we set $\partial e:=\{i,j\}$.
Then we obtain
\be\label{time:ord}
\TT_n=\sum_{\Gamma\in G_n}\TT_{\Gamma}
\ee
with $G_n$ the set of all graphs with vertices $\{1,\dots n\}$ and 
\(\TT_{\Gamma}=\frac{1}{\textrm{Sym}(\Gamma)}\langle\widetilde{S}_{\Gamma},\delta_{\Gamma}\rangle\)
where 
\[\widetilde{S}_{\Gamma}=\prod_{e\in E(\Gamma)}\Delta^F_S(x_{e,i},i\in\partial e)\]
\[\delta_{\Gamma}=\frac{\delta^{|E(\Gamma)|}}{\prod_{i\in V(\Gamma)}\prod_{e:i\in\partial e}\delta\ph_i(x_{e,i})}|_{\ph_1=\dots=\ph_n}\]
and the symmetry factor is $\textrm{Sym}(\Gamma)=\prod_{i<j}l_{ij}!$. Note that $\widetilde{S}_{\Gamma}$ is a well-defined distribution in $\Dcal'((\M^2\backslash\Diag )^{\left|E(\Gamma)\right|})$ ($\Diag$ denotes the thin diagonal) that can be uniquely extended to $\Dcal'(\M^{2\left|E(\Gamma)\right|})$,
since the Feynman fundamental solution has a unique extension with the same scaling degree.
 More explicitly we can write \eqref{time:ord} as:
\be\label{time:ord2}
\TT_n(F_1,\dots, F_n)=\sum_{\Gamma\in G_n}\frac{1}{\textrm{Sym}(\Gamma)}\langle\widetilde{S}_{\Gamma},\delta_{\Gamma}(F_1,\dots, F_n)\rangle
\ee

Graphically we represent $F$ with a vertex $\FD$ and   $D_{ij}$ with a dumbbell ${}_i\dumbbell{}_j$, so each empty circle corresponds to a functional derivative. Applying the derivative on a  functional can be pictorially represented as filling the circle with the vertex. Note that the expansion in graphs is possible due to the fact that the action used as a starting point is quadratic, so $D$ is a second order differential operator. If it were of order $k>2$, instead of lines we would have had to use $k-1$ simplices to represent it. Let us illustrate the concepts which we introduced here on a simple example.
\begin{exa}[removing tadpoles]
Let us look at the definition of the time ordered product of $F$ and $G$ in low orders in $\hbar$.
We can write $D(F\cdot G)$ diagramatically as:
\begin{align}
\frac{1}{i\hbar}D\left(F\cdot G\right) & =\left\langle \De_S^F,F^{\left(2\right)}\right\rangle G+F\left\langle \De_S^F,G^{\left(2\right)}\right\rangle +2\left\langle \De_S^F,F^{\left(1\right)}\otimes G^{\left(1\right)}\right\rangle \label{eq:TadpoleLeibnizRule}\\
 & =\FoneFG+\hspace{-1em}\FGoneG+2\FoneG\nonumber \end{align}
Here we see that the \textit{tadpoles} are present. The lowest order contributions to $e^{-\frac12D}F$ can be written as:
\[
e^{-\frac12D}F=F-\tfrac{1}{2}DF+\mathcal{O}(\hbar^{2})=\FD-\hbar\FFlop+\mathcal{O}(\hbar^{2})\,.
\]
Now we write the expression for $F\T G$ up to the first order in $\hbar$:
\[
\left(1+\tfrac{1}{2}D\right)\left[\left(1-\tfrac{1}{2}D\right)F\cdot\left(1-\tfrac{1}{2}D\right)G\right]=\FG+\hbar\FoneG+\mathcal{O}(\hbar^{2}).
\]
All the loop terms cancel out. We can see that applying $e^{-\frac12D}$ on $G$ and $F$ reflects what is called in physics ``removing the tadpoles''. In formula \eqref{expDij} it is reflected by the fact that we set $l_{ii}=0$.
\end{exa}

As long as the formula \eqref{expD} is applied to regular functionals there is no problem, since their functional derivatives are by definition test functions. But the relevant functionals are the interaction Lagrangians which are local functionals and therefore have derivatives with support on the thin diagonal, hence all but the first derivative are singular. As a typical example consider
\[F(\ph)=\int dz f(z)\frac{\ph(z)^k}{k!}\ .\] 
Its derivatives are
\be\label{integral:rep}
F^{(l)}[\ph](x_1,\dots,x_l)=\int dz f(z)\frac{\ph(z)^{k-l}}{(k-l)!}\prod_i\delta(z-x_i)\ .
\ee
In general, the functional derivatives of a local functional have the form
\[F^{(l)}[\ph](x_1,\dots,x_l)=\int dz \sum_jf_j[\ph](z)p_j(\partial_{x_1},\dots,\partial_{x_l})\prod_{i=1}^l\delta(z-x_i)\]
with polynomials $p_j$ and $\ph$-dependent test functions $f_j[\ph]$. The integral representation above is not unique since one can add total derivatives.
This amounts to the relation
\be\label{deriv:cond0}
\int dz q(\partial_z)f(z)p(\partial_{x_1},\dots,\partial_{x_l})\prod_i\delta(z-x_i)=\int dz f(z)q(\partial_{x_1}+\dots\partial_{x_l})p(\partial_{x_1},\dots,\partial_{x_l})\prod_i\delta(z-x_i)\ .
\ee

We insert the integral representation \eqref{integral:rep} into the formula \eqref{time:ord2} for the time ordered product and in each term we obtain:
\[
\langle\widetilde{S}_{\Gamma},\delta_{\Gamma}(F_1,\dots, F_n)\rangle=\int\!\! d\vec{x} d\vec{z}\smashoperator[r]{\prod\limits_{v\in V(\Ga)}}\Big( \sum\limits_{j_v}f^{v}_{j_v}[\ph](z_v)p_{j_v}(\partial_{x_{e,v}}|v\in\partial e)\prod\limits_{\mathclap{e: v\in\partial e}}^{\al_v}\de^4(z_v-x_{e,v})\Big)\widetilde{S}_{\Gamma}\,,
\]
where $\al_v$ is the number of lines adjacent at vertex $v$ and we use the notation $\vec{x}=(x_{e,v}|e\in E(\Ga),v\in\partial e)$, $\vec{z}=(z_v|v\in V(\Ga))$. 
We can move the partial derivatives $\partial_{x_{e,v}}$ by formal partial integration to the distribution $\widetilde{S}_{\Gamma}$. Next we integrate over the delta distributions, which  amounts to the pullback of a derivative of $\widetilde{S}_{\Ga}$ with respect to the 
map $\rho_{\Ga}:\M^{|V(\Ga)|}\rightarrow \M^{2|E(\Ga)|}$ given by the prescription
\[
(\rho_{\Ga}(z))_{e,v}=z_v\,\quad\mathrm{if}\,v\in\partial e\,.
\]
Let $p$ be a polynomial in the derivatives with respect to the partial derivatives $\partial_{x_{e,v}},v\in\partial e$. The pullback $\rho_{\Ga}^*$ of $p\widetilde{S}_{\Ga}$ is well defined on
$\M^{\left|V(\Gamma)\right|}\backslash\DIAG$, where $\DIAG$ is the large diagonal:
 \[
\DIAG=\left\{ z\in\M^{\left|V(\Gamma)\right|}|\,\exists v,w\in V(\Gamma),v\neq w:\, z_{v}=z_{w}\right\} \,.
\]
The problem of renormalization now amounts to finding the extensions of $\rho_{\Ga}^*p\widetilde{S}_{\Ga}$ to everywhere defined distributions $S_{\Ga,p}\in\Dcal'(\M^{\left|V(\Gamma)\right|})$
which depend linearly on $p$. These extensions must satisfy the relation
\be\label{deriv:cond}
\partial_{z_v}S_{\Ga,p}=S_{\Ga,(\sum_{e}\partial_{x_{e,v}})p}
\ee
We present now the inductive procedure of Epstein and Glaser that allows to define the desired extension of $\rho_{\Ga}^*p\widetilde{S}_{\Ga}$. For the simplicity of notation we first consider the case where no derivative couplings are present. 

Let us define an \emph{Epstein-Glaser subgraph (EG subgraph)} $\gamma\subseteq\Gamma$
to be a subset of the set of vertices $V(\gamma)\subseteq V(\Gamma)$
together with all lines in $\Gamma$ connecting them,\[
E(\gamma)=\left\{ e\in E(\Gamma):\partial e \subset V(\gamma)\right\} .\]
The first
step of the Epstein-Glaser induction is to choose extensions for all
EG subgraphs with two vertices, $\left|V(\gamma)\right|=2$. In this
case we have translation invariant distributions in $\Dcal'(\M^{2}\backslash\Diag)$, which correspond in relative coordinates to generic distributions  $\widetilde{t}_\gamma$  in
$\Dcal'(\M\backslash\left\{ 0\right\} )$. 
The scaling degree of these distributions is given by $\left|E(\gamma)\right|\left(d-2\right)$,
and we can choose a (possibly unique) extension according to Theorem~\ref{extension}.
By translation invariance this gives extensions $t_{\gamma}\in\Dcal'(\M^{2})$. 

Now we come to the induction step. For a generic EG subgraph $\gamma\subseteq\Gamma$ with $n$ vertices we assume
that the extensions of distributions corresponding to all EG subgraphs of $\gamma$ with less than $n$ vertices
have already been chosen. The causality condition \ref{causality}
then gives a translation invariant distribution in $\Dcal'(\M^{\left|V(\gamma)\right|}\backslash\Diag)$
which corresponds to a generic distribution $\widetilde{t}_{\gamma}\in\Dcal'(\M^{\left|V(\gamma)\right|-1}\backslash\left\{ 0\right\} )$.
The scaling degree and hence the degree of divergence of this distribution
is completely fixed by the structure of the graph:
\begin{equation}
\mathrm{div}(\gamma)=\left|E(\gamma)\right|\left(d-2\right)-\left(\left|V(\gamma)\right|-1\right)d\,,\qquad d=\dim(\M)\,.\label{eq:DegreeOfDivergenceGraph}\end{equation}
We call $\gamma$ superficially convergent if $\mathrm{div}(\gamma)<0$, logarithmically
divergent if  $\mathrm{div}(\gamma)=0$ and divergent of degree
$\mathrm{div}(\gamma)$ otherwise. Again by Theorem~\ref{extension}
there is a choice to be made in the extension of $\widetilde{t}_{\gamma}$
in the case $\mathrm{div}(\gamma)\geq0$.  

Let us now come back to the case where derivative couplings are present.
The scaling degree of $p\widetilde{S}_{\Ga}$ fulfills:
\[
\mathrm{sd}(p\widetilde{S}_{\Ga})\leq\mathrm{sd}(\widetilde{S_{\Gamma}})+|p|\,,
\]
where $|p|$ is the degree of the polynomial $p$. We can see that $p$ encodes the derivative couplings appearing in the graph $\gamma$. In the framework of Connes-Kreimer Hopf algebras it is called the \textit{external structure of the graph}. The presence of derivative couplings introduces an additional freedom in the choice of the extension in each step of the Epstein-Glaser induction and one has to use it to fulfill \eqref{deriv:cond}. This relation follows basically from the Action Ward Identity, as discussed in \cite{DF07,DF02}. It can be also seen as a consistency condition implementing the Leibniz rule, see \cite{HW5}.

Let us now remark on the relation of the  Epstein-Glaser induction  to a more conventional approach to renormalization. Firstly we show, how the EG renormalization relates to the regularization procedure.  We are given an EG subgraph $\gamma$ with $n$ vertices and we assume that all the subgraphs with $n-1$ vertices are already renormalized. Let
\be
\Dcal_\lambda(\M^{n-1}):=\{f\in\Dcal(\M^{n-1})\,|\,(\partial^\alpha f)(0)=0\,\,\,\forall |\alpha|\leq\lambda\}
\ee
be the space of functions with derivatives vanishing up to order $\lambda$
and let $\Dcal'_\lambda(\M^{n-1})$ be the corresponding space of distributions. Theorem \ref{extension} tells us that the distribution  $\widetilde{t}_{\gamma}\in\Dcal'(\M^{n-1})$ associated with the EG subgraph $\gamma$ has a unique extension to an element of  $\Dcal'_{\mathrm{div}(\gamma)}(\M^{n-1})$. An extension to a distribution on the full space $\Dcal(\M^{n-1})$ can be therefore defined by a choice of the projection:
\[
W:\Dcal(\M^{n-1})\rightarrow \Dcal_{\mathrm{div}(\gamma)}(\M^{n-1})\,.
\]
There is a result proven in \cite{DF04}, which characterizes all such projections:
\begin{prop}\label{lem:W-Projection-Functions}
There is a one-to-one correspondence between families of functions
\begin{equation}
\left\{ w_{\alpha}\in\Dcal\,|\quad\forall\left|\beta\right|\leq\lambda:\partial^{\beta}w_{\alpha}(0)=\delta_{\alpha}^{\beta},\,\left|\alpha\right|\leq\lambda\right\} \label{eq:LemmaProjectionWFunctions}\end{equation}
and projections $W:\Dcal\rightarrow\Dcal_{\lambda}$. The set (\ref{eq:LemmaProjectionWFunctions})
defines a projection $W$ by
\be\label{W:proj}
Wf:=f-\sum_{\left|\alpha\right|\leq\lambda}f^{\left(\alpha\right)}(0)\, w_{\alpha}\,.
\ee
Conversely a set of functions of the form (\ref{eq:LemmaProjectionWFunctions})
is given by any basis of $\mathrm{ran}(1-W)$ dual to the basis $\left\{ \delta^{\left(\alpha\right)}:\left|\alpha\right|\leq\lambda\right\} $
of $\Dcal_{\lambda}^{\perp}\subset\Dcal'$.
\end{prop}
Let us now define, following \cite{Kai}, what we mean by a regularization of a distribution.
\begin{df}[Regularization]\label{df:regularisation} Let $\tilde{t}\in\Dcal'(\RR^n\setminus\{0\})$ be a distribution with degree of 
divergence $\lambda$, and let $\bar{t}\in\Dcal_\lambda'(\RR^n)$ be the unique extension of $\tilde{t}$ with the same degree of divergence. A family of distributions $\{t^\zeta\}_{\zeta\in\Omega\setminus\{0\}}$, $t^\zeta\in\Dcal'(\RR^n)$, with $\Omega\subset\CC$ a neighborhood of the origin, is called a regularization of $\tilde{t}$, if
\be\label{eq:regularization}
\forall g\in\Dcal_\lambda(\RR^n):\quad\lim_{\zeta\rightarrow0}\langle t^\zeta,g\rangle=\langle \bar{t},g\rangle\,.
\ee
The regularization $\{t^\zeta\}$ is called analytic, if for all functions $f\in\Dcal(\RR^n)$ the map
\be
\Omega\setminus\{0\}\ni\zeta\mapsto \langle t^\zeta,f \rangle
\ee
is analytic with a pole of finite order at the origin. The regularization $\{t^\zeta\}$ is called finite, if 
the limit $\lim_{\zeta\rightarrow 0}\langle t^\zeta,f\rangle\in\CC$ exists $\forall f\in\Dcal(\RR^n)$; 
in this case $\lim_{\zeta\rightarrow0}t^\zeta\in\Dcal'(\RR^n)$ is called an extension or renormalization of $\tilde{t}$.
\end{df}

For a finite regularization the limit
$\lim_{\zeta\rightarrow0}t^\zeta$ is indeed a solution $t$ of the
extension problem. Given a regularization $\{t^\zeta\}$ of $t$, it follows from \eqref{eq:regularization} 
that for any projection $W:\Dcal\rightarrow\Dcal_\lambda$
\be\label{regW-1}
\langle \bar{t},Wf\rangle=\lim_{\zeta\rightarrow0}\langle t^\zeta,Wf\rangle\, \quad \forall f\in\Dcal(\RR^n)\,.
\ee
Any extension $t\in\Dcal'(\RR^n)$ of $\tilde t$ with the same scaling degree is of the form $\langle t,f\rangle=\langle \bar t,Wf\rangle$ with some $W$-projection of the form \eqref{W:proj}.
Since $t^\zeta\in\Dcal'(\RR^n)$ we can write (\ref{regW-1}) in the form
\be\label{regW-2}
\langle \bar{t},Wf\rangle=\lim_{\zeta\rightarrow0}\left[\langle t^\zeta,f\rangle - \sum_{|\al|\leq\sd(t)-n}\langle t^\zeta,w_\al\rangle\; f^{(\al)}(0)\right].
\ee
In general the limit on the right hand side cannot be split, since the limits of the individual terms might not exist. However, if the regularization $\{t^\zeta,\zeta\in\Omega\setminus\{0\}\}$ is analytic, each term can be expanded in a Laurent series around $\zeta=0$, and since the overall limit is finite, the principal parts ($\pp$) of these Laurent series must coincide.
It follows that the principal part of any analytic regularization $\{t^\zeta\}$ of a distribution $t\in\Dcal'(\RR^n\setminus\{0\})$ is a local distribution of order $\sd(t)-n$. We can now give a definition of the minimal subtraction in the EG framework.
\begin{cor}[Minimal Subtraction]\label{cor:MS-same-sd}
The regular part ($\rp=1-\pp$) of any analytic regularization $\{t^\zeta\}$ of a distribution $\tilde{t}\in\Dcal'(\RR^n\setminus\{0\})$ defines by
\be\label{def:MS}
\langle t^\MS,f\rangle :=\lim_{\zeta\rightarrow0} \rp(\langle t^\zeta,f\rangle)
\ee
an extension of $\tilde{t}$ with the same scaling degree, $\sd(t^\MS)=\sd(\tilde{t})$.
The extension $t^\MS$ defined by (\ref{def:MS}) is called ``minimal subtraction''.
\end{cor}
To finish this discussion we want to remark on the difference between the Epstein-Glaser procedure and the BPHZ scheme. It is best seen on the example of the rising sun diagram of the $\ph^4$ theory.  In the framework of BPHZ, it contains three logarithmically divergent subdiagrams, which have to be renormalized first. In the perspective of EG, however, it is a diagram with two vertices and, hence, contains no divergent subdiagram at all. This way one saves some work computing contributions, which, as shown by Zimmermann \cite{Z} cancel out in the end.

We have just seen how to define the $n$-fold time-ordered products (i.e. multilinear maps $\TT_n$) by the procedure of Epstein and Glaser. An interesting question is whether the renormalized time ordered product defined by such a sequence of multilinear maps can be understood as an iterated binary product on a suitable domain. 
Recently we proved in  \cite{FRQ} that this is indeed the case. The crucial observation is that multiplication of local functionals is injective. More precisely, let $\mathfrak{F}_0(M)$ be the set of local functionals vanishing at some distinguished field configuration (say $\ph=0$). Iterated multiplication $m$ is then a linear map from the symmetric Fock space over $\mathfrak{F}_0(M)$ onto the algebra of functionals which is generated by $\mathfrak{F}_0(M)$. Then there holds the following assertion:
\begin{prop}The multiplication $m:S^\bullet\mathfrak{F}_0(M)\to\F(M)$ is bijective (where $S^k$ denotes the symmetrised tensor product of vector spaces).
\end{prop} 
Let $\beta=m^{-1}$. We now define the renormalized time ordering operator on the space of multilocal functionals $\F(M)$ by
\be
\TTR:=(\bigoplus_n  \TT_n)\circ\beta 
\ee  
This operator is a formal power series in $\hbar$ starting with the identity, hence it is injective.
The renormalized time ordered product is now defined on the image of $\TTR$ by
\be
A\TRH B\doteq\TTR(\TTR^{\minus}A\cdot\TTR^{\minus}B)\,,
\ee
This product is equivalent to the pointwise product and is in particular associative and commutative. Moreover, the $n$-fold time ordered product of local functionals coincides with the $n$-linear map $\TT_n$ of causal perturbation theory.
\clearpage
\appendix
\section{Distributions and wavefront sets}\label{distr}
We recall same basic notions from the theory of distributions on $\RR^n$. Let $\Omega\subset\RR^n$ be an open subset and $\Ecal(\Omega)\doteq\Ci(\Omega,\RR)$ the space of smooth functions on it. We equip this space with a Fr\'echet topology generated by the family of seminorms:
\be\label{topE}
p_{K,m}(\ph)=\sup_{x\in K\atop |\alpha|\leq m}|\partial^\alpha\ph(x)|\,,
\ee
where $\alpha\in\NN^N$ is a multiindex and $K\subset \Omega$ is a compact set. This is just the topology of uniform convergence on compact sets, of all the derivatives. 

The space of smooth compactly supported functions $\Dcal(\Omega)\doteq\Ci_c(\Omega,\RR)$ can be equipped with a locally convex topology in a similar way. The fundamental system of seminorms is given by \cite{Sch0}:
\be\label{topD}
p_{\{m\},\{\epsilon\},a}(\ph)=\sup_\nu\big(\sup_{|x|\geq\nu,\atop |p|\leq m_\nu} \big|D^p\ph^a(x)\big|/\epsilon_\nu\big)\,,
\ee
where $\{m\}$ is an increasing sequence of positive numbers going to $+\infty$ and $\{\epsilon\}$ is a decreasing one tending to $0$. 

The space of \textbf{\textit{distributions}}\index{distribution} is defined to be the dual  $\Dcal'(\Omega)$ of $\Dcal(\Omega)$ with respect to the topology given by (\ref{topD}). Equivalently, given a linear map $L$ on $\Dcal(\Omega)$ we can decide if it is a distribution by checking one of the equivalent conditions given in the theorem below \cite{Trev,Rud,Hoer}.
\begin{thm}
A linear map $u$ on $\Ecal(\Omega)$ is a distribution if it satisfies the following equivalent conditions:
\begin{enumerate}
\item To every compact subset $K$ of $\Omega$ there exists an integer $m$ and a constant $C>0$ such that for all $\ph\in\Dcal$ with support contained in $K$ it holds:
\[
|u(\ph)|\leq C\max_{p\leq k}\sup_{x\in\Omega}|\pa^p\ph(x)|\,.
\]
We call $||u ||_{\Ccal^k(\Omega)}\doteq\max_{p\leq k}\sup_{x\in\Omega}|\pa^p\ph(x)|$ the $\Ccal^k$-norm and if  the same integer $k$ can be used in all $K$ for a given distribution $u$, then we say that $u$ is of order $k$\index{order of a distribution}.
\item If a sequence of test functions $\{\ph_k\}$, as well as all their derivatives converge uniformly to 0 and if all the test functions $\ph_k$ have their supports contained in a compact subset $K\subset\Omega$ independent of the index $k$, then $u(\ph_k)\rightarrow 0$.
\end{enumerate}
\end{thm}
An important property of a distribution is its support\index{distribution!support}.  If $U' \subset U$ is an open subset then $\Dcal(U')$ is a closed subspace of $\Dcal(U)$ and there is a natural restriction map $\Dcal'(U) \rightarrow \Dcal'(U')$. We denote the restriction of a distribution $u$ to an open subset $U'$ by $u|_{U'}$.
\begin{df}
The support $\supp u$ of a distribution $u \in \Dcal'(\Omega)$ is the smallest closed set $\Ocal$ such that $u|_{\Omega\setminus \Ocal} = 0$. In other words:
\[
\supp u\doteq \{x\in\Omega|\, \forall U\,\textrm{open neigh. of }x,\, U\subset\Omega\ \exists \ph\in\Dcal(\Omega), \supp\ph\subset U,\,\mathrm{s.t. }<\!u,\ph\!>\neq 0\}\,.
\]
\end{df}
Distributions with compact support\index{distribution!with compact support} can be characterized by means of a following theorem:
\begin{thm}
The set of distributions in $\Omega$ with compact support is identical with the dual $\Ecal'(\Omega)$ of $\Ecal(\Omega)$ with respect to the topology given by (\ref{topE}).
\end{thm}

Now we discuss the singularity structure of distributions. This is mainly based on \cite{Hoer} and chapter 4 of \cite{BaeF}.
\begin{df}
The singular support\index{singular support} $\mathrm{sing\, supp}\, u$ of $u \in \Dcal'(\Omega)$ is the smallest closed subset $\Ocal$ such that $u|_{\Omega\setminus \Ocal} \in \Ecal(\Omega\setminus \Ocal)$.
\end{df}
We recall an important theorem giving the criterium for a compactly distribution to have an empty singular support:
\begin{thm}
A distribution $u \in \Ecal'(\Omega)$ is smooth if and only if for every $N$ there is a constant $C_N$ such that:
\[
|\hat{u}({k} )| \leq C_N (1 + |{k} |)^{-N}\,,
\]
where $\hat{u}$ denotes the Fourier transform of $u$.
\end{thm}
We can see that a distribution is smooth if its Fourier transform decays fast at infinity.
If a distribution has a nonempty singular support we can give a further characterization of its singularity structure by specifying the direction in which it is singular. This is exactly the purpose of the definition of a wave front set.
\begin{df}
For a distribution $u \in \Dcal'(\Omega)$ the wavefront set\index{wavefront set} $\WF(u)$ is the complement in $\Omega \times \RR^n\setminus\{0\}$ of the set of points $(x,{k}) \in \Omega \times \RR^n\setminus\{0\}$ such that there exist
\begin{itemize}
\item a function $f \in\Dcal(\Omega)$ with $f(x)=1$,
\item an open conic neighborhood $C$ of ${k}$, with
\[
\sup_{{k}\in C}(1+|{k}|)^N|\widehat{f \cdot u}({k})|<\infty\qquad\forall N \in \NN_0\,.
\]
\end{itemize}
\end{df}

On a manifold $M$ the definition of the Fourier transform depends on the choice of a chart, but the property of strong decay in some direction (characterized now by a point $(x, k)$, $k\neq 0$ of the cotangent bundle $T^*M$) turns out to be independent of this choice. Therefore the wave front set WF of a distribution on a manifold $M$ is a well defined closed conical subset of the cotangent bundle (with the zero section removed).

The wavefront sets provide a simple criterion for the existence of point-wise products of distributions. Before we give it, we prove a more general result concerning the pullback. Here we follow closely \cite{BaeF,Hoer}. Let $F:X\rightarrow Y$ be a smooth map between $X\subset \RR^m$ and $Y\subset \RR^n$. We define the  normal set $N_F$ of the map $F$ as:
\[
N_F\doteq \{(F(x),\eta)\in Y\times \RR^n| ((dF_x)^T(\eta)=0\}\,,
\] 
where $(dF_x)^T$ is the transposition of the differential of $F$ at x.
\begin{thm}
Let $\Gamma$ be a closed cone in  $Y\times (\RR^n\{0\})$ and $F:X\rightarrow Y$ as above, such that $N_F\cap\Gamma=\varnothing$. Then the pullback of functions $F^*:\Ecal(X)\rightarrow\Ecal(Y)$ has a unique, sequentially continuous extension to
a sequentially continuous map $\Dcal'_\Gamma(Y)\rightarrow\Dcal'(X)$, where $\Dcal'_\Gamma(Y)$ denotes the space of distributions with WF sets contained in $\Gamma$.
\end{thm}
\begin{proof} Here we give only an idea of the proof. Details can be found in  \cite{BaeF,Hoer}. Firstly, one has to show that the problem can be reduced to a local construction. Let $x\in X$. We assumed that $N_F\cap \Gamma=\varnothing$, so we can choose a compact neighborhood $K$ of $F(x)$ and an open neighborhood $\Ocal$ of $x$ such that $\overline{F(\Ocal)} \subset \textrm{int}(K)$ and the following condition holds:
\[
\exists \epsilon>0\ \textrm{s.t.}\ V\doteq\overline{\bigcup\limits_{x\in\Ocal}\{{k}|(dF_x)^T{k}\}}\ \textrm{satisfies }(K\times V)\cap\Gamma=\varnothing\,.
\]
Such neighborhoods define an cover of  $X$ and we choose  its locally finite refinement which we denote by $\{\Ocal_\al\}_{\al\in A}$, where $A$ is some index set. To this cover we have the associated family of compact sets $K_\la\subset Y$ and we choose a partition of unity $\sum\limits_{\al\in A}g_\al=1$, $\supp g_\al\subset\Ocal_\al$ and a family $\{f_\al\}_{\al\in A}$ of functions on $Y$ with $\supp f_\al=K_\al$ and $f_\al\equiv 1$ on $F(\supp g_\al)$. Then:
\[
F^*(\ph) = \sum\limits_{\al\in A} g_\al F^*(f_{\al}\ph)\,.
\]
This way the problem reduces to finding an extension of $F^*_\al\doteq(F\big|_{\Ocal_\al})^* : \Ci_c(K_\al,\RR) \rightarrow \Ci(\Ocal_\al,\RR)$
to a map on $\Dcal'_\Gamma(K_\al)$. Note that for $\ph\in C_c\infty(K_\al)$, $\supp \chi\subset\Ocal_\al$, we can write the pullback as:
\[
\left<F_\al^*(\ph),\chi\right>=\int \ph(F_\al(x))\chi(x) dx=\int \hat{\ph}(\eta)e^{i\left<F_\al(x),\eta\right>}\chi(x) dxd\eta=\int\hat{\ph}(\eta)T_\chi(\eta) d\eta\,,
\]
where we denoted $T_\chi(\eta)\doteq\int e^{i\left<F_\al(x),\eta\right>}\chi(x)dx$. We can use this expression to define the pullback for  $u\in\Dcal'_\Gamma(K_\al)$, by setting:
\[
\left<F_\al^*(u),\chi\right>\doteq\int\hat{u}(\eta)T_\chi(\eta) d\eta\,.
\]
To show that this integral converges, we can divide it into two parts: integration over $V_\al$ and over $\RR^n\setminus V_\al$, i.e.:
\[
\left<F_\al^*(u),\chi\right>=\int\limits_{V_\al}\hat{u}(\eta)T_\chi(\eta) d\eta\,+\int\limits_{\RR^n\setminus V_\al}\hat{u}(\eta)T_\chi(\eta) d\eta\\,.
\]
The first integral converges since $K_\al\times V_\al\cap\Gamma=\varnothing$ and therefore $\hat{u}(\eta)$ decays rapidly on $V_\al$, whereas $|T_\chi(\eta)|\leq\int|\chi(x)|dx$. The second integral also converges. To prove it, first we note that $\hat{u}(\eta)$ is polynomially bounded i.e. $\hat{\ph}(\eta)\leq C(1+|\eta|)^N$ for some $N$ and appropriately chosen constant $C$. Secondly, we have a following estimate on $T_\chi(\eta)$: for ever $k\in\NN$ and a closed conic subset $V\subset \RR^n$ such that $(dF_x)^T\eta\neq 0$ for $\eta\in V$, there exists a constant $C_{k,V}$ for which it holds\footnote{For the proof of this estimate see \cite{BaeF,Hoer}}
\[
|T_\chi (\eta)| \leq C_{k,V} (1 + |\eta|)^{-k}\,,
\]
Since for $\eta\in V_\al$ it holds $(dF_x)^T\eta>\epsilon>0$, we can use this estimate to prove the convergence of the second integral.

We already proved that $F^*:\Dcal'_\Gamma(Y)\rightarrow\Dcal'(X)$ exists. Now it remains to show its sequential continuity. This can be easily done, with the use of estimates provided above and  the uniform boundedness principle.
\end{proof}

Using this theorem we can define the pointwise product of two distributions $t,s$ on an $n$-dimensional manifold $M$ as a pullback by the diagonal map $D:M\rightarrow M\times M$ if  the pointwise sum of their wave front sets
\[
\WF(t) + \WF(s) = \{(x, k + k')|(x, k) \in \WF(t), (x, k') \in \WF(s)\}\,,
\]
does not intersect the zero section of $\dot{T}^*M$. This is the theorem  8.2.10 of \cite{Hoer}.  To see that this is the right criterium, note that the set of normals of the diagonal map $D: x\mapsto (x,x)$ is given by $N_D=\{(x,x,k,-k)|x\in M, k\in T^*M\}$. 
The product $ts$ is be defined by: $ts=D^*(t\otimes s)$ and if one of $t,s$ is compactly supported, then so is $ts$ and we define the contraction by $\left<t,s\right>\doteq\widehat{ts}(0)$.
 \end{fmffile}

\end{document}